\documentclass[journal,onecolumn,11pt,twoside]{IEEEtran}
\usepackage{epsfig,graphics,subfigure}
\usepackage{amsmath,amssymb,amsthm,multirow,balance,dsfont,hyperref}
\usepackage{cite,cleveref}

\newtheorem{theorem}{Theorem}
\newtheorem{lemma}{Lemma}

\newtheorem{assumption}{Assumption}

\usepackage{color}

\newcommand{\bc}{\boldsymbol{c}}
\newcommand{\bs}{\boldsymbol{s}}
\newcommand{\bu}{\boldsymbol{u}}
\newcommand{\bv}{\boldsymbol{v}}
\newcommand{\bw}{\boldsymbol{w}}
\newcommand{\bz}{\boldsymbol{z}}
\newcommand{\bx}{\boldsymbol{x}}
\newcommand{\by}{\boldsymbol{y}}

\newcommand{\bpsi}{\boldsymbol{\psi}}
\newcommand{\bphi}{\boldsymbol{\phi}}

\newcommand{\bH}{\boldsymbol{H}}

\newcommand{\cA}{\mathcal{A}}
\newcommand{\cB}{\mathcal{B}}
\newcommand{\cC}{\mathcal{C}}
\newcommand{\cD}{\mathcal{D}}
\newcommand{\cE}{\mathcal{E}}
\newcommand{\cF}{\mathcal{F}}
\newcommand{\cS}{\mathcal{S}}
\newcommand{\cG}{\mathcal{G}}
\newcommand{\cH}{\mathcal{H}}

\newcommand{\cJ}{\mathcal{J}}

\newcommand{\cL}{\mathcal{L}}

\newcommand{\cN}{\mathcal{N}}

\newcommand{\cY}{\mathcal{Y}}
\newcommand{\cV}{\mathcal{V}}
\newcommand{\cX}{\mathcal{X}}

\newcommand{\cw}{{\scriptstyle\mathcal{W}}}

\newcommand{\bcB}{\boldsymbol{\cal{B}}}
\newcommand{\bcH}{\boldsymbol{\cal{H}}}
\newcommand{\bcF}{\boldsymbol{\cal{F}}}
\newcommand{\bcw}{\boldsymbol{\cw}}

\newcommand{\bwt}{\widetilde \bw}

\newcommand{\cwt}{\widetilde\cw}
\newcommand{\bcwt}{\widetilde\bcw}

\newcommand{\cwb}{\overline{\cw}}
\newcommand{\wb}{\overline{w}}

\newcommand{\expec}{\mathbb{E}}

\newcommand{\col}{\text{col}}
\newcommand{\bvc}{\text{bvec}}
\newcommand{\vc}{\text{vec}}
\newcommand{\tr}{\text{Tr}}
\newcommand{\diag}{\text{diag}}

\begin{document}


\title{{Learning} over Multitask Graphs --\\ {Part II: Performance Analysis}}
\author{{\normalsize{Roula Nassif, \IEEEmembership{Member, IEEE}, Stefan Vlaski, \IEEEmembership{Member, IEEE}, \\
C\'edric Richard, \IEEEmembership{Senior Member, IEEE}, Ali H. Sayed, \IEEEmembership{Fellow Member, IEEE}}}\\
\thanks{The work of A. H. Sayed was supported in part by NSF grants CCF-1524250 and ECCS-1407712. {A short  version of this work appeared in the conference publication~\cite{nassif2018distributed}.}

This work was done while R. Nassif was a post-doc at EPFL. She is now with the American University of Beirut, Lebanon (e-mail: roula.nassif@aub.edu.lb). S. Vlaski and A. H. Sayed are with Institute
of Electrical Engineering, EPFL, Switzerland (e-mail: stefan.vlaski,ali.sayed@epfl.ch). C. Richard is with Universit\'e de Nice Sophia-Antipolis, France (e-mail: cedric.richard@unice.fr).
}
}

\maketitle

\begin{abstract}
Part I of this paper  formulated a multitask optimization problem where agents in the network have individual objectives to meet, or individual parameter vectors to estimate, subject to a smoothness condition over the graph. 
A diffusion strategy {was} devised that responds to streaming data and employs stochastic approximations in place of actual gradient vectors, which are generally unavailable. The approach {relied} on minimizing a global cost consisting of the aggregate sum of individual costs regularized by a term that promotes smoothness. {We examined the first-order, the second-order, and the fourth-order stability of the multitask learning algorithm. The results identified conditions on the step-size parameter, regularization strength, and data characteristics in order to ensure stability. This Part II examines steady-state performance of the strategy. The results reveal} explicitly the influence of the network topology and the regularization strength on the network performance and provide insights into the design of effective multitask strategies for distributed inference over networks.
\end{abstract}


\newpage
\section{Introduction}
{As pointed out in Part I~\cite{nassif2018diffusion} of this work, most prior literature on distributed inference over networks focuses on single-task problems, where agents with separable objective functions need to agree on a common parameter vector corresponding to the minimizer of an aggregate sum of individual costs~\cite{bertsekas1997new,olfati2007consensus,nedic2009distributed,dimakis2010gossip,ram2010distributed,chen2013distributed,sayed2014adaptation,chen2015learning,chen2015learning2,sayed2014adaptive,vlaski2016diffusion}. In this paper, and its accompanying Part I~\cite{nassif2018diffusion}, we focus instead on multitask networks where the agents may need to estimate and track multiple objectives simultaneously~\cite{platachaves2017heterogeneous,hassani2017multi,chen2014multitask,nassif2016proximal,eksin2012distributed,hallac2015network,szurley2015distributed,platachaves2015distributed,alghunaim2017decentralized,nassif2017diffusion,chen2015diffusion,chen2014diffusion,kekatos2013distributed}. Although agents may generally have distinct though related tasks to perform, they may still be able to capitalize on inductive transfer between them to improve their performance. Based on the type of prior information that may be available about how the tasks are related to each other, multitask learning algorithms can be derived by translating the prior information into constraints on the parameter vectors to be inferred. }

{In Part I~\cite{nassif2018diffusion}, we considered multitask inference problems where each agent in the network seeks to minimize an individual cost expressed as the expectation of some loss function. The minimizers of the individual costs are assumed to vary smoothly over the topology, as dictated by the graph Laplacian matrix. The smoothness property softens the transitions in the tasks among adjacent nodes and allows incorporating information about the graph structure into the solution of the inference problem. In order to exploit the smoothness prior, we formulated the inference problem as the minimization of the aggregate sum of individual costs regularized by a term promoting smoothness, known as the graph-Laplacian regularizer~\cite{zhou2004regularization,shuman2013emerging}. A diffusion strategy was proposed that responds to streaming data and employs stochastic approximations in place of actual gradient vectors, which are generally unavailable. }

{The analysis from Part I~\cite{nassif2018diffusion} revealed how the regularization strength $\eta$ can steer the convergence point of the network toward many modes of operation starting from the non-cooperative mode ($\eta=0$) where each agent converges to the minimizer of its individual cost and ending with the single-task mode ($\eta\rightarrow\infty$) where all agents converge to a common parameter vector corresponding to the minimizer of the aggregate sum of individual costs. For any values of $\eta$ in the range $0<\eta<\infty$, the network behaves in a multitask mode where agents seek their individual models while at the same time ensuring that these models satisfy certain smoothness and closeness conditions dictated by the value of $\eta$.  We carried out in Part I~\cite{nassif2018diffusion} a detailed stability analysis of the proposed strategy. We showed, under conditions on the step-size learning parameter $\mu$, that the adaptive strategy induces a contraction mapping and that despite gradient noise, it is able to converge in the mean-square-error sense within $O(\mu)$ from the solution of the regularized problem, for sufficiently small $\mu$. We also established the first and fourth-order moments stability of the network error process and showed that they tend asymptotically to bounded region on the order of $O(\mu)$ and $O(\mu^2)$, respectively. }

{Based on the results established in Part I~\cite{nassif2018diffusion}, we shall derive in this paper a closed-form expression for the steady-state network mean-square-error relative to  the minimizer of the regularized cost. This closed form expression will reveal explicitly the influence of the regularization strength, network topology (through the eigenvalues and eigenvectors of the Laplacian matrix), gradient noise, and data characteristics, on the network performance. Additionally, a closed-form expression for the steady-state network mean-square-error relative to  the minimizers of the individual costs is also derived. This expression will provide insights on the design of effective multitask strategies for distributed inference over networks.}

\noindent\textbf{Notation.} {We adopt the same notation from Part I~\cite{nassif2018diffusion}.} All vectors are column vectors. Random quantities are denoted in boldface. Matrices are denoted in capital letters while vectors and scalars are denoted in lower-case letters. The operator $\preceq$ denotes an element-wise inequality; i.e., $a\preceq b$ implies that each entry of the vector $a$ is less than or equal to the corresponding entry of $b$. The symbol $\diag\{\cdot\}$ forms a matrix from block arguments by placing each block immediately below and to the right of its predecessor. The operator $\col\{\cdot\}$ stacks the column vector entries on top of each other. The symbols $\otimes$ and $\otimes_b$ denote the Kronecker product and the block Kronecker product, respectively. The symbol $\vc(\cdot)$ refers to the standard vectorization operator that stacks the columns of a matrix on top of each other and the symbol $\bvc(\cdot)$ refers to the block vectorization operation that vectorizes each block and stacks the vectors on top of each other.

\section{Distributed inference under smoothness priors}

\subsection{Problem formulation and adaptive strategy}
Consider a connected network (or graph) $\cG=\{\cN,\cE,A\}$, where $\cN$ is a set of $N$ agents (nodes), $\cE$ is a set of edges connecting agents with particular relations, and  $A$ is a symmetric, weighted adjacency matrix. If there is an edge connecting agents $k$ and $\ell$, then $[A]_{k\ell}=a_{k\ell}>0$ reflects the strength of the relation between $k$ and $\ell$; otherwise, $[A]_{k\ell}=0$. We introduce the graph Laplacian $L$, which is a differential operator defined as $L=D-A$, where the degree matrix $D$ is a diagonal matrix with $k$-th entry $[D]_{kk}=\sum_{\ell=1}^Na_{k\ell}$. Since $L$ is symmetric positive semi-definite, it possesses a complete set of orthonormal eigenvectors. We denote them by $\{v_1,\ldots,v_N\}$. For convenience, we order the set of real, non-negative eigenvalues of $L$ as $0=\lambda_1<\lambda_2\leq\ldots\leq\lambda_N=\lambda_{\max}(L)$, where, since the network is connected, there is only one zero eigenvalue with corresponding eigenvector $v_1=\frac{1}{\sqrt{N}}\mathds{1}_N$~\cite{chung1997spectral}. Thus, the Laplacian can be decomposed as:
\begin{equation}
\label{eigendecomposition of the Laplacian}
L=V\Lambda V^\top,
\end{equation}
where $\Lambda=\diag\{\lambda_1,\ldots,\lambda_N\}$ and $V=[v_1,\ldots,v_N]$.

Let $w_k\in\mathbb{R}^M$ denote some parameter vector at agent $k$ and let $\cw=\col\{w_1,\ldots,w_N\}$ denote the collection of parameter vectors from across the network. We associate with each agent $k$ a risk function $J_k(w_k):\mathbb{R}^M\rightarrow\mathbb{R}$ assumed to be strongly convex. In most learning and adaptation problems, the risk function is expressed as the expectation of a loss function $Q_k(\cdot)$ and is written as $J_k(w_k)=\expec\,Q_k(w_k;\bx_k)$, where $\bx_k$ denotes the random data. The expectation is computed over the distribution of this data. We denote the unique minimizer of $J_k(w_k)$ by $w^o_k$.  {{Let us recall the assumption on the risks $\{J_k(w_k)\}$ used in Part I~\cite{nassif2018diffusion}}.
\begin{assumption}{\rm(Strong convexity)}
\label{assumption: strong convexity}
 It is assumed that the individual costs $J_k(w_k)$ are each twice differentiable and strongly convex such that the Hessian matrix function $H_k(w_k)=\nabla^2_{w_k}J_k(w_k)$ is uniformly bounded from below and {above, say, as:}
\begin{equation}
0<\lambda_{k,\min}I_M\leq H_k(w_k)\leq\lambda_{k,\max}I_M,
\end{equation}
where $\lambda_{k,\min}>0$ for $k=1,\ldots,N$.
\qed
\end{assumption}
}

In many situations, there is prior information available about $\cw^o=\col\{w_1^o,\ldots,w_N^o\}$. In the current {Part II, and its accompanying Part I~\cite{nassif2018diffusion},} the prior belief we want to enforce is that the target signal $\cw^o$ is smooth with respect to the underlying weighted graph. References~\cite{chen2014multitask,nassif2016proximal} provide variations for such problems for the special case of mean-square-error costs. {Here we treat general convex costs. Let} $\cL= L\otimes I_M$. The smoothness of $\cw$ can be measured in terms of a quadratic form of the graph {Laplacian~\cite{zhou2004regularization,shuman2013emerging,ando2006learning,dong2016learning,chen2017bias}}:
\begin{equation}
\label{eq: quadratic regularization}
S(\cw)=\cw^\top\cL \cw=\frac{1}{2}\sum_{k=1}^N\sum_{\ell\in\cN_k}a_{k\ell}\|w_k-w_{\ell}\|^2,
\end{equation}
where $\cN_k$ is the set of neighbors of $k$, i.e., the set of nodes connected to agent $k$ by an edge. The smaller $S(\cw)$ is, the  smoother the  signal $\cw$ on the  graph is. Intuitively, given that the weights are non-negative, $S(\cw)$ shows that $\cw$ is  considered  to  be  smooth if nodes with a large $a_{k\ell}$ on the edge connecting them have similar weight values $\{w_k,w_{\ell}\}$. Our objective is to devise and study a strategy that solves the following regularized problem:
\begin{equation}
\label{eq: global problem}
\cw^o_{\eta}=\arg\min_{\cw}J^{\text{glob}}(\cw)=\sum_{k=1}^NJ_k(w_k)+\frac{\eta}{2}\, \cw^{\top}\cL \cw,
\end{equation}
in a distributed manner where each agent is interested in estimating the $k$-th sub-vector of $\cw^o_{\eta}=\col\{w_{1,\eta}^o,\ldots,w_{N,\eta}^o\}$. The tuning parameter $\eta\geq 0$ controls the trade-off between the two components  of the objective function. 
We are particularly interested in solving the problem in the stochastic setting when the distribution of the data $\bx_k$ in  $J_k(w_k)=\expec\,Q_k(w_k;\bx_k)$ is generally unknown. This means that the risks $J_k(w_k)$ and their gradients $\nabla_{w_k}J_k(w_k)$ are unknown. As such, approximate gradient vectors need to be employed. A common construction in stochastic approximation theory is to employ the following approximation at  iteration $i$:
\begin{equation}
\widehat{\nabla_{w_k}J_k}(w_k)=\nabla_{w_k}Q_k(w_k;\bx_{k,i}),
\end{equation}
where $\bx_{k,i}$ represents the data observed at iteration $i$. The difference between the true gradient and its approximation is called the gradient noise $\bs_{k,i}(\cdot)$:
\begin{equation}
\label{eq: gradient noise process}
\bs_{k,i}(w)\triangleq\nabla_{w_k}J_k(w)-\widehat{\nabla_{w_k}J_k}(w).
\end{equation}
{Let $\bw_{k,i}$ denote the estimate of $w^o_{k,\eta}$ at iteration $i$ and node $k$. In order to solve~\eqref{eq: global problem} in a fully distributed and adaptive manner, we proposed in Part I~\cite{nassif2018diffusion} the following diffusion-type algorithm}:
\begin{equation}
\label{eq: distributed algorithm}
\left\lbrace
\begin{array}{lr}
\bpsi_{k,i}=\bw_{k,i-1}-\mu\widehat{\nabla_{w_k}J_k}(\bw_{k,i-1})\\
\bw_{k,i}=\bpsi_{k,i}-\mu\eta\displaystyle\sum_{\ell\in\cN_k}a_{k\ell}(\bpsi_{k,i}-\bpsi_{\ell,i}),
\end{array}
\right.
\end{equation}
where $\mu>0$ is a small step-size parameter and $\bpsi_{k,i}$ is an intermediate variable.

\subsection{{Summary of main results}}
{One key insight that followed from the analysis in Part I~\cite{nassif2018diffusion} is that the smoothing parameter $\eta$ can be regarded as an effective tuning parameter that controls the nature of the learning process. The value of $\eta$ can vary from $\eta=0$ to $\eta\rightarrow\infty$. We showed that at one end, when $\eta=0$, the learning algorithm reduces to a non-cooperative mode of operation where each agent acts individually and estimates its own local model, $w^o_k$. On the other hand, when $\eta\rightarrow\infty$, the learning algorithm moves to a single-mode of operation where all agents cooperate to estimate a {\em single} parameter (namely, the Pareto solution of the aggregate cost function). For any values of $\eta$ in the range $0<\eta<\infty$, the network behaves in a multitask mode where agents seek their individual models while at the same time ensuring that these models satisfy certain smoothness and closeness conditions dictated by the value of $\eta$. }

{In Part I~\cite{nassif2018diffusion}, we carried out a detailed stability analysis of the proposed strategy~\eqref{eq: distributed algorithm}. We showed, under some conditions on the step-size parameter $\mu$, that:
\begin{align}
\limsup_{i\rightarrow\infty}\|\expec(\cw^o_{\eta}-\bcw_i)\|&=O(\mu),~\quad\text{(see~\cite[Theorem 4]{nassif2018diffusion})}\label{eq: mean stability}\\
\limsup_{i\rightarrow\infty}\expec\|\cw^o_{\eta}-\bcw_i\|^2&=O(\mu), ~\quad\text{(see~\cite[Theorem 2]{nassif2018diffusion})}\label{eq: mean-square stability}\\
\limsup_{i\rightarrow\infty}\expec\|\cw^o_{\eta}-\bcw_i\|^4&=O(\mu^2),\quad\text{(see~\cite[Theorem 3]{nassif2018diffusion})}\label{eq: mean-fourth stability}
\end{align}
where $\cw^o_\eta$ is the solution of the regularized problem~\eqref{eq: global problem} and $\bcw_i=\col\{\bw_{1,i},\ldots,\bw_{N,i}\}$ denotes the network block weight vector at iteration $i$. Expression~\eqref{eq: mean-square stability} indicates that the mean-square error $\expec\|\cw^o_{\eta}-\bcw_i\|^2$ is on the order of $\mu$. However, in this Part II, we are interested in characterizing how close the $\bcw_i$ gets to  the network limit point $\cw^o_{\eta}$. In particular, we will be able to characterize the network mean-square deviation (MSD) (defined below in~\eqref{eq: network MSD performance}) value  in terms of the step-size $\mu$, the regularization strength $\eta$, the network topology (captured by the eigenvalues $\lambda_m$ and eigenvectors $v_m$ of the Laplacian $L$), and the data characteristics (captured by the second-order properties of the costs $H_{k,\eta}$ and second-order moments of the gradient noise $R_{s,k,\eta}$) as follows:
\begin{equation}
\label{eq: final expression for the network MSD alternative intro}
\text{MSD}=\frac{\mu}{2N}\sum_{m=1}^N{\tr}\left(\left(\sum_{k=1}^N[v_m]^2_kH_{k,\eta}+\eta\lambda_mI\right)^{-1}\left(\sum_{k=1}^N[v_m]^2_kR_{s,k,\eta}\right)\right)+\frac{O(\mu)}{(O(1)+O(\eta))},
\end{equation}
where $[v_m]_k$ denotes the $k$-th entry of the eigenvector $v_m$. The interpretation of~\eqref{eq: final expression for the network MSD alternative intro} is explained in more detail in Section~\ref{sec: mean-square performance} where it is shown, by  coupling $\eta$ and $\mu$ in an appropriate manner, that the term $\frac{O(\mu)}{(O(1)+O(\eta))}$ will be a strictly higher order term of $\mu$. As we will explain later in Sections~\ref{sec: mean-square performance} and~\ref{sec: multitask learning benefit}, by properly setting the parameters, expression~\eqref{eq: final expression for the network MSD alternative intro} allows us to recover the mean-square-deviation of stand-alone adaptive agents ($\eta=0$) and single-task diffusion networks ($\eta\rightarrow\infty$).   }

{Recall that the objective of the multitask strategy~\eqref{eq: distributed algorithm} is to exploit similarities among neighboring agents in an attempt to improve the overall network performance in approaching the collection of individual minimizer ${\cw}^o$ by means of local communications. Section~\ref{sec: multitask learning benefit} in this paper is devoted to quantify the \emph{benefit} of cooperation, namely, the objective of improving the mean-square deviation around the limiting point of the algorithm. In particular, we will be able to characterize the mean-square-deviation ($\overline{\text{MSD}}$) value relative to the multitask objective \( \cw^o \) in terms of the MSD in~\eqref{eq: final expression for the network MSD alternative intro} and the mismatch $\cw^o_\eta-\cw^o$ as follows:
\begin{equation}
\label{eq: final expression for the network MSD relative to wo alternative intro}
\overline{\text{MSD}}=\underbrace{\text{MSD}}_{O(\mu),\eta}+\underbrace{\|\cw^o_\eta-\cw^o\|^2}_{\text{smoothness},\eta}+2(\cw^o_\eta-\cw^o)^\top\underbrace{\text{bias}}_{O(\mu),\eta},
\end{equation}
where ``bias'' is the bias of algorithm~\eqref{eq: distributed algorithm} relative to $\cw^o_{\eta}$ given in future expression~\eqref{eq: steady-state bias of long model}. By increasing $\eta$, the MSD in~\eqref{eq: final expression for the network MSD alternative intro} is more likely to decrease. However, by increasing $\eta$, from expression~(31) in Part I~\cite{nassif2018diffusion}, $\|\cw^o_\eta-\cw^o\|^2$ is more likely to increase and the size of this increase is determined by the smoothness of $\cw^o$. From future Lemma~\ref{lemma: mean stability long term}, it turns out that the third term on the RHS in~\eqref{eq: final expression for the network MSD relative to wo alternative intro} is a function of $\mu$, $\eta$, and the smoothness of the multitask objective $\cw^o$. By increasing $\eta$, this term is more likely to increase. The key conclusion will be that, while the second and third terms on the RHS in~\eqref{eq: final expression for the network MSD relative to wo alternative intro} will  in general increase as the regularization strength \( \eta \) increases, the {size} of this increase is determined by the smoothness of $\cw^o$ which is in turn  function of the network topology  captured by \( L \). The more similar the tasks at neighboring agents are, the smaller these terms will be. This implies that as long as  $\cw^o$ is sufficiently smooth, moderate regularization strengths \( \eta \) in the range $]0,\infty[$ exist such that $\overline{\text{MSD}}$ at these values of $\eta$ will be less than $\overline{\text{MSD}}$ at $\eta=0$ which corresponds to the non-cooperative mode of operation. The best choice for $\eta$ would be the one minimizing $\overline{\text{MSD}}$ in~\eqref{eq: final expression for the network MSD relative to wo alternative intro}. We refer the reader to Fig.~2 in~\cite[Section~II-B]{nassif2018diffusion} for an illustration of this concept of multitask learning benefit. This example will be considered further in the numerical experiments section.}
\subsection{Modeling Assumptions from Part I~\cite{nassif2018diffusion}}
{In this section, we recall the assumptions used in Part I~\cite{nassif2018diffusion} to establish the network mean-square error stability~\eqref{eq: mean-square stability}.
\begin{assumption}
\label{assumption: gradient noise}
\emph{(Gradient noise process)} The gradient noise process defined in~\eqref{eq: gradient noise process} satisfies for any $\bw\in\bcF_{i-1}$ {and for all $k,\ell=1,2,\ldots,N$:
\begin{align}
\expec[\bs_{k,i}(\bw)|\bcF_{i-1}]&=0,\label{eq: mean gradient noise condition}\\
\expec[\|\bs_{k,i}(\bw)\|^2|\bcF_{i-1}]&\leq\beta^2_k\|\bw\|^2+\sigma^2_{s,k},\label{eq: condition on second-order moment of gradient noise}\\
\expec[\bs_{k,i}(\bw)\bs_{\ell,i}^\top(\bw)|\bcF_{i-1}]&=0,\quad k\neq \ell,\label{eq: uncorrelated gradient noises}
\end{align}
for some} $\beta^2_k\geq 0$,  $\sigma^2_{s,k}\geq 0$, and where $\bcF_{i-1}$ denotes the filtration generated by the random processes $\{\bw_{\ell,j}\}$ for all $\ell=1,\ldots,N$ and $j\leq i-1$.
\qed
\end{assumption}}
%
\noindent {Let us introduce the network block vector $\bcw_i=\col\{\bw_{1,i},\ldots,\bw_{N,i}\}$. Recall from Part I~\cite{nassif2018diffusion} that at each iteration, we can view~\eqref{eq: distributed algorithm} as a mapping from $\bcw_{i-1}$ to $\bcw_i$:
\begin{equation}
\label{eq: network vector recursion}
\boxed{\bcw_i=\left(I_{MN}-\mu\eta\cL\right)\left(\bcw_{i-1}-\mu\,\col\left\{\widehat{\nabla_{w_k}J_k}(\bw_{k,i-1})\right\}_{k=1}^N\right)}
\end{equation}
{We introduced the following condition} on the combination matrix $(I_{MN}-\mu\eta\cL)$.
\begin{assumption}
\label{assumption: combination matrix}
\emph{(Combination matrix)} The symmetric combination matrix $\left(I_{MN}-\mu\eta\cL\right)$ {has nonnegative entries and its spectral radius is equal to one. Since $L$ has an eigenvalue at zero, these} conditions are satisfied when the step-size $\mu>0$ and the regularization strength $\eta\geq 0$ satisfy:
\begin{align}
\label{eq: condition for stability}&0\leq\mu\eta\leq\frac{2}{\lambda_{\max}(L)},\\
\label{eq: condition for positivity}&0\leq\mu\eta\leq\min_{1\leq k\leq N}\left\{\frac{1}{\sum_{\ell=1}^Na_{k\ell}}\right\},
\end{align}
where condition~\eqref{eq: condition for stability} ensures stability and condition~\eqref{eq: condition for positivity} ensures non-negative entries.
\qed
\end{assumption}}

{The results in Part I~\cite{nassif2018diffusion} established that the iterates $\bw_{k,i}$ converge in the mean-square-error sense to a small $O(\mu)-$ neighborhood around the regularized solution $w^o_{k,\eta}$. In this part of the work, we will be more precise and determine the size of this neighborhood, i.e., assess the size of the constant multiplying $\mu$ in the $O(\mu)-$term. To do so, we shall derive an accurate first-order expression for the mean-square error~\eqref{eq: mean-square stability}; the expression will be accurate to first-order in $\mu$.}

{To arrive at the desired expression, we first need to introduce a long-term approximation model and assess how close it is to the actual model. We then derive the performance for the long-term model and use this closeness to transform this result into an accurate expression for the performance of the original learning algorithm. 
To derive the long-term model, we follow the approach developed in~\cite{sayed2014adaptation}. The first step is to establish the asymptotic stability of the  fourth-order moment of the error vector, $\expec\|\cw_{\eta}^o-\bcw_{i}\|^4$, which has already been done in Part I~\cite{nassif2018diffusion}. This property is needed to justify the validity of the long-term approximate model. Recall that to establish the fourth-order stability, we replaced condition~\eqref{eq: condition on second-order moment of gradient noise} on the gradient noise process by the following condition on its fourth order moment:
\begin{equation}
\label{eq: condition on fourth-order moment of gradient noise}
\expec\left[\|\bs_{k,i}(\bw_k)\|^4|\bcF_{i-1}\right]\leq\overline{\beta}_k^4\|\bw_k\|^4+\overline{\sigma}_{s,k}^4,
\end{equation}
for some $\overline{\beta}_k^4\geq 0$,  and $\overline{\sigma}_{s,k}^4\geq 0$. As explained in~\cite{sayed2014adaptation}, condition~\eqref{eq: condition on fourth-order moment of gradient noise} implies~\eqref{eq: condition on second-order moment of gradient noise}.}

{To establish the mean-stability~\eqref{eq: mean stability}, we introduced a smoothness condition on the Hessian matrices of the individual costs. This smoothness condition will be adopted in the next section when we study the long term behavior of the network.
\begin{assumption}{\emph{(Smoothness condition on individual cost functions).}}
\label{assumption: local smoothness hessian}
It is assumed that each $J_k(w_k)$ satisfies a smoothness condition close to $w^o_{k,\eta}$, in that the corresponding Hessian matrix is Lipchitz continuous in the proximity of $w^o_{k,\eta}$ with some parameter $\kappa_d\geq 0$, i.e.,
\begin{equation}
\label{eq: hessian smoothness condition}
\|\nabla^2_{w_k}J_k(w^o_{k,\eta}+\Delta w_k)-\nabla^2_{w_k}J_k(w^o_{k,\eta})\|\leq \kappa_d\|\Delta w_k\|,
\end{equation}
for small perturbations $\|\Delta w_k\|\leq\epsilon$.
\qed
\end{assumption}}

\section{Long-term network dynamics}
\label{sec: Long-term network dynamics}
Let $\bcwt_i=\cw^o_{\eta}-\bcw_i$. Subtracting the vector $(I_{MN}-\mu\eta\cL)\cw^o_{\eta}$ from both sides of recursion~\eqref{eq: network vector recursion}, and using~\eqref{eq: gradient noise process}, we obtain:
\begin{equation}
\label{eq: error recursion wt big relation}
\bcwt_{i}-\mu\eta\cL\cw^o_{\eta}=(I_{MN}-\mu\eta\cL)\left(\bcwt_{i-1}+\mu\,\col\left\{\nabla_{w_k}J_k(\bw_{k,i-1})-\bs_{k,i}(\bw_{k,i-1})\right\}_{k=1}^{N}\right),
\end{equation}
From the mean-value theorem~\cite[pp.~24]{polyak1987introduction},\cite[Appendix~D]{sayed2014adaptation}, we have:
\begin{equation}
\label{eq: nabla w k at w k}
\nabla_{w_k}J_k(\bw_{k,i-1})=\nabla_{w_k}J_k(w^o_{k,\eta})-\bH_{k,i-1}(w^o_{k,\eta}-\bw_{k,i-1}),
\end{equation}
where
\begin{equation}
\bH_{k,i-1}\triangleq\int_{0}^1\nabla^2_{w_k}J_k(w^o_{k,\eta}-t(w^o_{k,\eta}-\bw_{k,i-1}))dt,
\end{equation}
{and from the optimality condition of~\eqref{eq: global problem}, we have:
\begin{equation}
\label{eq: optimality condition}
\col\left\{\nabla_{w_k}J_k(w_{k,\eta}^o)\right\}_{k=1}^N=-\eta\cL \cw^o_{\eta}.
\end{equation}}
Replacing~\eqref{eq: nabla w k at w k} into~\eqref{eq: error recursion wt big relation} and using~\eqref{eq: optimality condition}, we arrive at the following recursion for $\bcwt_i$:
\begin{equation}
\label{eq: error recursion wt}
\bcwt_{i}=(I_{MN}-\mu\eta\cL)(I_{MN}-\mu\bcH_{i-1})\bcwt_{i-1}-\mu(I_{MN}-\mu\eta\cL)\bs_{i}(\bcw_{i-1})+\mu^2\eta^2\cL^2\cw^o_{\eta},
\end{equation}
where
\begin{align}
\bs_{i}(\bcw_{i-1})&\triangleq\col\{\bs_{k,i}(\bw_{k,i-1})\}_{k=1}^N,\\
\bcH_{i-1}&\triangleq\diag\{\bH_{k,i-1}\}_{k=1}^N.
\end{align}

We move on to motivate a long-term model for the evolution of the network error dynamics, $\bcwt_{i}$, after sufficient iterations, i.e., for $i\gg1$. We examine the stability property of the model, the proximity of its trajectory {to} that of the original network {dynamics~\eqref{eq: error recursion wt}}, and subsequently employ the model to assess network performance. 

 \subsection{Long-term error model}
We introduce the error matrix $\widetilde{\bcH}_{i-1}\triangleq\cH_{\eta}-\bcH_{i-1}$, which measures the deviation of $\bcH_{i-1}$ from the constant matrix:
\begin{equation}
\label{eq: H of eta hessian}
\cH_{\eta}\triangleq\diag\{H_{k,\eta}\}_{k=1}^N,
\end{equation} with each $H_{k,\eta}$ given by the value of the Hessian matrix at the {regularized solution}, namely, $H_{k,\eta}\triangleq\nabla_{w_k}^2J_k(w_{k,\eta}^o)$. Let
 \begin{align}
 \bcB_{i-1}&\triangleq(I_{MN}-\mu\eta\cL)(I_{MN}-\mu\bcH_{i-1}),\\
 \cB_{\eta}&\triangleq(I_{MN}-\mu\eta\cL)(I_{MN}-\mu\cH_{\eta}).\label{eq: matrix B_eta}
 \end{align}
Then, we can write:
\begin{equation}
\label{eq: B_i-1 0}
\bcB_{i-1}=\cB_{\eta}+\mu(I_{MN}-\mu\eta\cL)\widetilde{\bcH}_{i-1}.
\end{equation}
Using~\eqref{eq: B_i-1 0}, we can rewrite the error recursion~\eqref{eq: error recursion wt} as:
\begin{equation}
\label{eq: error recursion wt 0}
\begin{split}
\bcwt_{i}=\cB_{\eta}\bcwt_{i-1}-\mu(I_{MN}-\mu\eta\cL)\bs_{i}(\bcw_{i-1})+\mu^2\eta^2\cL^2\cw^o_{\eta}+\mu(I_{MN}-\mu\eta\cL)\bc_{i-1},
\end{split}
\end{equation}
in terms of the random perturbation sequence:
\begin{equation}
\label{eq: definition of c i-1}
\bc_{i-1}\triangleq\widetilde{\bcH}_{i-1}\bcwt_{i-1}.
\end{equation}
Under Assumptions~\ref{assumption: strong convexity} and~\ref{assumption: local smoothness hessian}, and for small $\mu$, it can be shown that  $\limsup_{i\rightarrow\infty}\expec\|\bc_{i-1}\|=O(\mu)$, and that $\|\bc_{i-1}\|=O(\mu)$ asymptotically with \textit{high probability} {(see Appendix~\ref{app: perturbation sequence size})}. {Motivated by this result, we introduce the following approximate model, where the last term involving $\bc_{i-1}$ in~\eqref{eq: error recursion wt 0}, which is $O(\mu^2)$, is removed:}
\begin{equation}
\label{eq: long term error model 1}
\begin{split}
\bcwt'_{i}=\cB_{\eta}\bcwt'_{i-1}-\mu(I_{MN}-\mu\eta\cL)\bs_{i}(\bcw_{i-1})+\mu^2\eta^2\cL^2\cw^o_{\eta}, \quad i\gg1.
\end{split}
\end{equation}
Obviously, the iterates that are generated by~\eqref{eq: long term error model 1} are generally different from the iterates generated by the original recursion~\eqref{eq: error recursion wt}. To highlight this fact, we {are using} the prime notation for the state of the long-term model. Note that the driving process $\bs_{i}(\bcw_{i-1})$ in~\eqref{eq: long term error model 1}  is the {{\em same}} gradient noise process from the original recursion~\eqref{eq: error recursion wt} and is evaluated at $\bcw_{i-1}$. In the following, we show that, after sufficient iterations $i\gg 1$, the error dynamics of the network relative to the {solution} $\cw^o_{\eta}$ is well-approximated by the  model~\eqref{eq: long term error model 1}.
\subsection{Size of Approximation Error}
{We start by showing} that the mean-square difference between the trajectories $\{\bcwt_i,\bcwt'_i\}$ is asymptotically bounded by $O(\mu^2)$ and that the {mean-square error performance} of the long term model~\eqref{eq: long term error model 1}  is within $O(\mu^{\frac{3}{2}})$ from the {performance} of the original recursion~\eqref{eq: error recursion wt}. Working with recursion~\eqref{eq: long term error model 1} is much more tractable for performance analysis because its dynamics is driven by the constant matrix $\cB_{\eta}$ as opposed to the random matrix $\bcB_{i-1}$ in the original error recursion~\eqref{eq: error recursion wt}. Therefore, we shall work with the long-term model~\eqref{eq: long term error model 1} and evaluate its {performance}, which {will provide} an accurate representation for the {performance} of the original distributed strategy~\eqref{eq: distributed algorithm} to first order in the step-size $\mu$.

\begin{lemma}{\emph{(Size of approximation error)}}
\label{lemma: dimension of the approximation error}
Under Assumptions~\ref{assumption: strong convexity},~\ref{assumption: gradient noise},~\ref{assumption: combination matrix}, and~\ref{assumption: local smoothness hessian}, and condition~\eqref{eq: condition on fourth-order moment of gradient noise}, it holds that:
 \begin{align}
 \limsup_{i\rightarrow\infty}\expec\|\bcwt_i-\bcwt_i'\|^2&=O(\mu^2),\\
 \limsup_{i\rightarrow\infty} \expec\|\bcwt_i\|^2&= \limsup_{i\rightarrow\infty}\expec\|\bcwt_i'\|^2+O(\mu^{\frac{3}{2}}).
 \end{align}
\end{lemma}
\begin{proof}
See Appendix~\ref{app: proof of dimension of the approximation error}.
\end{proof}
\noindent We shall discuss now the mean and mean-square error stability of the long-term approximate model~\eqref{eq: long term error model 1}.
 \subsection{Stability of First-Order Error Moment}
Conditioning both sides of~\eqref{eq: long term error model 1}, invoking the conditions on the gradient noise from Assumption~\ref{assumption: gradient noise}, and computing the conditional expectations, we obtain:
\begin{equation}
\expec[\bcwt'_{i}|\bcF_{i-1}]=\cB_{\eta}\bcwt'_{i-1}+\mu^2\eta^2\cL^2\cw^o_{\eta}.
\end{equation}
Taking expectation again, we arrive at:
\begin{equation}
\label{eq: mean recursion of the long term model}
\expec\bcwt'_{i}=\cB_{\eta}\expec\bcwt'_{i-1}+\mu^2\eta^2\cL^2\cw^o_{\eta}.
\end{equation}
The above recursion is stable if the matrix $\cB_{\eta}$ in~\eqref{eq: matrix B_eta} is stable. This matrix has a similar form {to} the matrix $(I_{MN}-\mu\eta\cL)(I_{MN}-\mu\cH_{\infty})$ encountered {in Part I~\cite[Section III-A2]{nassif2018diffusion}}. Similarly, it can be verified that $\cB_{\eta}$ is stable when {condition~\eqref{eq: condition for stability}  and condition
\begin{equation}
\label{eq: condition 1}
0<\mu<\min_{1\leq k\leq N}\left\{\frac{2}{\lambda_{k,\max}}\right\}.
\end{equation} are satisfied. }In this case, we {obtain}
\begin{equation}
\label{eq: steady-state bias of long model}
\cwt'_{\infty}\triangleq\lim_{i\rightarrow\infty}\expec\bcwt'_{i}=\mu^2\eta^2(I_{MN}-(I_{MN}-\mu\eta\cL)(I_{MN}-\mu\cH_{\eta}))^{-1}\cL^2\cw^o_{\eta},
\end{equation}
where the {RHS} in the above expression is similar to the {RHS} {in equation (49)~\cite{nassif2018diffusion}} with $\cH_{\infty}$ replaced by $\cH_{\eta}$.

\begin{lemma}{\emph{(Mean stability of long-term model)}}
\label{lemma: mean stability long term}
Under Assumptions~\ref{assumption: strong convexity},~\ref{assumption: gradient noise},~\ref{assumption: combination matrix}, and~\ref{assumption: local smoothness hessian}, and for sufficiently small $\mu$, the steady-state bias $\cwt'_{\infty}=\lim_{i\rightarrow\infty}\expec\bcwt'_{i}$ of the long-term model~\eqref{eq: long term error model 1} given by~\eqref{eq: steady-state bias of long model} satisfies:
\begin{equation}
\label{eq: bound on the bias '}
{\mu\lim_{\mu\rightarrow 0}\left(\frac{1}{\mu}\lim_{i\rightarrow\infty}\|\expec\bcwt'_{i}\|\right)\leq \mu\frac{O(\eta^2)}{(O(1)+O(\eta))^2}.}
\end{equation}
\end{lemma}
\begin{proof}
The proof is similar to the proof of {Theorem 1 in Part I~\cite{nassif2018diffusion}} 
 with $\cH_{\infty}$  replaced by $\cH_{\eta}$.
\end{proof}
\subsection{Stability of Second-Order Error Moment}
In the following, we show that the  long term approximate model~\eqref{eq: long term error model 1} is also mean-square stable in the sense that $\expec\|\bwt'_{k,i}\|^2$ tends asymptotically to a region that is bounded by $O(\mu)$. We follow the same line of reasoning as {in Part I~\cite[Section III-A]{nassif2018diffusion}} where we studied the mean-square stability of the original model~\eqref{eq: error recursion wt}. Based on the inequality:
\begin{align}
\limsup_{i\rightarrow\infty}\expec\|\cw^o_{\eta}-\bcw'_i\|^2&=\limsup_{i\rightarrow\infty}\expec\|\cw^o_{\eta}-\cw'_{\infty}+\cw'_{\infty}-\bcw'_i\|^2\notag\\
&\leq2\|\cw^o_{\eta}-\cw'_{\infty}\|^2+2\limsup_{i\rightarrow\infty}\expec\|\cw'_{\infty}-\bcw'_i\|^2,\label{eq: bound on mean-square expectation long term}
\end{align}
where $\cw^o_{\eta}-\cw'_{\infty}=\cwt'_{\infty}$ is the steady-state bias of the long term model given by~\eqref{eq: steady-state bias of long model} and where $\cw'_{\infty}-\bcw'_i$ follows the recursion:
\begin{equation}
\label{eq: recursion for w' infinity w'i}
\cw'_{\infty}-\bcw'_{i}=\cB_{\eta}(\cw'_{\infty}-\bcw'_{i-1})-\mu(I_{MN}-\mu\eta\cL)\bs_{i}(\bcw_{i-1}),
\end{equation}
and from {Theorems 1 and 2 in~Part I~\cite{nassif2018diffusion} and previous} Lemma~\ref{lemma: mean stability long term}, we can establish the mean-square stability of~\eqref{eq: long term error model 1}. Let us introduce the mean-square perturbation vector ($\text{MSP}'$) at time $i$ relative to $\cw'_{\infty}$:
\begin{equation}
\text{MSP}'_i\triangleq\col\left\{\expec\|w'_{k,\infty}-\bw'_{k,i}\|^2\right\}_{k=1}^N.
\end{equation}

\begin{lemma}{\emph{(Mean-square stability of the long-term model)}}
\label{lemma: mean-square stability long term}
Under Assumptions~\ref{assumption: strong convexity},~\ref{assumption: gradient noise},~\ref{assumption: combination matrix}, and~\ref{assumption: local smoothness hessian}, the $\emph{MSP}'$ at time $i$ can be recursively bounded as:
\begin{equation}
\label{eq: evolution of the MSP' i}
\emph{MSP}'_i \preceq (I_N-\mu\eta L)(G'')^2\emph{MSP}'_{i-1}+3\mu^2 (I_N-\mu\eta L)\emph{\diag}\{\beta^2_k\}_{k=1}^N\emph{MSP}_{i-1}+\mu^2 (I_N-\mu\eta L)b,
\end{equation}
{where:
\begin{align}
G''&\triangleq\emph{\diag}\left\{\gamma_k\right\}_{k=1}^N,\label{eq: equation of G''}\\
b&\triangleq{\emph\col}\left\{\sigma^2_{s,k}+3\beta_k^2\|w^o_{k,\eta}\|^2+ 3\beta_k^2\|w^o_{k,\eta}-w_{k,\infty}\|^2\right\}_{k=1}^N,\\
\gamma_k&\triangleq\max\{|1-\mu\lambda_{k,\min}|,|1-\mu\lambda_{k,\max}|\}.\label{eq: gamma_k}
\end{align}
and $\emph{MSP}_{i}$ is the mean-square perturbation vector at time $i$ relative to the fixed point $\cw_\infty=\emph{\col}\{w_{k,\infty}\}_{k=1}^N$ of algorithm~\eqref{eq: distributed algorithm} in the absence of gradient noise (see~\cite[Section III-A3]{nassif2018diffusion}).}  A sufficiently small $\mu$ ensures the stability of the above recursion. It follows that
\begin{equation}
\label{eq: steady-state of the MSP'}
\|\limsup_{i\rightarrow\infty}{\emph{MSP}}'_i\|_{\infty}=O(\mu),
\end{equation}
and that 
\begin{equation}
\label{eq: steady-state of the second order long term}
\limsup_{i\rightarrow\infty}\expec\|\bcwt'_i\|^2=O(\mu)+\frac{O(\mu^2\eta^4)}{(O(1)+O(\eta))^{4}}= O(\mu).
\end{equation}
\end{lemma}
\begin{proof}
See Appendix~\ref{app: proof of mean-square stability of the long term model}.
\end{proof}

\section{Mean-square-error performance}
\label{sec: mean-square performance}
We established in {Theorem~2 in Part I~\cite{nassif2018diffusion}} that a network running strategy~\eqref{eq: distributed algorithm} is mean-square-error stable for sufficiently small $\mu$. {Specifically,} we showed that $\limsup_{i\rightarrow\infty}\expec\|\cw^o_{\eta}-\bcw_{i}\|^2=O(\mu)$. In the following, we assess the size of the network mean-square-deviation (MSD) using the definition~\cite[Chapter 11]{sayed2014adaptation}:
\begin{equation}
\label{eq: network MSD performance}
\text{MSD}\triangleq{\mu\lim_{\mu\rightarrow 0}\left(\limsup_{i\rightarrow\infty}\frac{1}{\mu}\,\expec\left(\frac{1}{N}\|\cw^o_{\eta}-\bcw_{i}\|^2\right)\right)}.
\end{equation}
In addition to Assumptions~\ref{assumption: strong convexity},~\ref{assumption: gradient noise},~\ref{assumption: combination matrix},~\ref{assumption: local smoothness hessian}, and condition~\eqref{eq: condition on fourth-order moment of gradient noise} on the individual costs, $J_k(w_k)$, the gradient noise process, $\bs_{k,i}(\bw_k)$, and the combination matrix, $I_N-\mu\eta L$, we introduce a smoothness condition on the noise covariance matrices.

For any $\bw_k\in\bcF_{i-1}$, we let
\begin{equation}
\label{eq: R ski definition}
R_{s,k,i}(\bw_k)\triangleq\expec[\bs_{k,i}(\bw_k)\bs_{k,i}^\top(\bw_k)|\bcF_{i-1}]
\end{equation}
denote the conditional second-order moment of the gradient noise process, which generally depends on $i$ because the statistical distribution of $\bs_{k,i}(\bw_k)$ can be iteration-dependent, and is random since it depends on the random iterate $\bw_k$. We assume that, in the limit, this covariance matrix tends to a constant value when evaluated at $w^o_{k,\eta}$ and we denote the limit by:
\begin{equation}
\label{eq: R sk definition}
R_{s,k,\eta}\triangleq\lim_{i\rightarrow\infty}\expec[\bs_{k,i}(w^o_{k,\eta})\bs_{k,i}^\top(w^o_{k,\eta})|\bcF_{i-1}].
\end{equation}
\begin{assumption}{\emph{(Smoothness condition on the noise covariance)}}
\label{assumption: Smoothness condition on noise covariance}
It is assumed that the conditional second-order moment of the noise process is locally Lipschitz continuous in a small neighborhood around $w^o_{k,\eta}$, namely,
\begin{equation}
\label{eq: smoothness condition on noise covariance}
\|R_{s,k,i}(w^o_{k,\eta}+\Delta w_k)-R_{s,k,i}(w^o_{k,\eta})\|\leq\kappa_d\|\Delta w_k\|^{{\theta}},
\end{equation}
for small perturbations $\|\Delta w_k\|\leq\epsilon$, and for some constant $\kappa_d\geq 0$ and exponent $0<{\theta}\leq 4$.
\qed
\end{assumption}
\noindent One useful conclusion that follows from Assumption~\ref{assumption: Smoothness condition on noise covariance} is that, after sufficient iterations, we can express the covariance matrix of the gradient noise process, $\bs_{k,i}(\bw_k)$, in terms of the limiting matrix $R_{s,k,\eta}$ defined in~\eqref{eq: R sk definition}. Specifically, following the same proof used to establish Lemma 11.1 in~\cite{sayed2014adaptation}, we can show that under the smoothness condition~\label{eq: smoothness condition on noise covariance} and for small step-size, the covariance matrix of the gradient noise process, $\bs_{k,i}(\bw_{k,i-1})$, at each agent $k$ satisfies for $i\gg 1$:
\begin{equation}
\label{eq: relation on the covariance}
\expec\bs_{k,i}(\bw_{k,i-1})\bs_{k,i}^{\top}(\bw_{k,i-1})=R_{s,k,\eta}+O\left(\mu^{\min\left\{1,\frac{{{\theta}}}{2}\right\}}\right).
\end{equation}
{For clarity of presentation, we sketch the proof in Appendix~\ref{app: proof of limiting second order of gradient noise} where we used results from Theorems~2 and~3 in Part I~\cite{nassif2018diffusion}.}

Before studying the steady-state network performance, we establish some properties of the matrix:
\begin{equation}
\cF_{\eta}\triangleq\cB_{\eta}^\top\otimes_b\cB_{\eta}^\top,\label{eq: matrix F_eta}
\end{equation}
which is defined in terms of the block Kronecker operation using blocks of size $M\times M$. {In the derivation that follows,} we shall use the block Kronecker product $\otimes_b$ operator~\cite{koning1991block} and the block vectorization operator $\bvc(\cdot)$. As explained in~\cite{sayed2014adaptation}, these {operations} preserve the locality of the blocks in the original matrix arguments.  Since $\rho(\cF_{\eta})=\left(\rho(\cB_{\eta})\right)^2$, the matrix $\cF_{\eta}$ is stable under conditions~\eqref{eq: condition for stability}  and~\eqref{eq: condition 1}. This matrix plays a critical role in characterizing the performance of the distributed multitask algorithm. In our derivations, the matrix $\cF_\eta$ will also appear transformed under the orthonormal transformation:
\begin{equation}
\label{eq: transformed matrix F_eta}
\overline{\cF}_{\eta}\triangleq(\cV\otimes_b\cV)^{\top}\cF_{\eta}(\cV\otimes_b\cV),
\end{equation}
{where $\cV\triangleq V\otimes I_M$.}
\begin{lemma}{\emph{(Coefficient matrix $\cF_{\eta}$)}}
\label{lemm: coefficient matrix}
For sufficiently small step-size, it holds that
\begin{equation}
(I-\cF_{\eta})^{-1}=O(\mu^{-1}),
\end{equation}
and
\begin{align}
\label{eq: relation inverse I-F}
(I-\overline{\cF}_{\eta})^{-1}&=X^{-1}+W\\
&=\mu^{-1}\cdot\left[\begin{array}{cc}
O(1)&0\\
0&(O(1)+O(\eta))^{-1}\\
\end{array}\right]+\mu^{-1}\cdot\left[\begin{array}{cc}
(O(1)+O(\eta))^{-1}&(O(1)+O(\eta))^{-1}\\
(O(1)+O(\eta))^{-1}&(O(1)+O(\eta))^{-2}\\
\end{array}\right]
\end{align}
where $X$ is an $N^2\times N^2$ block diagonal matrix with each block of dimension $M^2\times M^2$:
\begin{equation}
X=\mu\cdot\emph{\diag}\left\{\emph{\diag}\left\{(1-\mu\eta\lambda_m)(1-\mu\eta\lambda_p)(H_{mm}\oplus H_{pp})+\eta(\lambda_m+\lambda_p-\mu\eta\lambda_m\lambda_p){I_{M^2}}\right\}_{p=1}^N\right\}_{m=1}^N,
\end{equation}
with $\oplus$ denoting the Kronecker sum operator~\cite{bernstein2005matrix}:
\begin{align}
H_{mm}\oplus H_{pp}=H_{mm}\otimes I_M+I_M\otimes H_{pp},\label{eq: H m p plus}\\
H_{mn}\triangleq(v_m^{\top}\otimes I_M)\cH_\eta(v_n\otimes I_M).\label{eq: H mn}
\end{align}
The matrix $W$ is an $N^2\times N^2$ block matrix arising from the matrices $\{H_{mn}| m\neq n\}$. Moreover, we can also write:
\begin{equation}
(I-\cF_{\eta})^{-1}=(\cV\otimes_b\cV)X^{-1}(\cV\otimes_b\cV)^{\top}+\mu^{-1}(O(1)+O(\eta))^{-1}.\label{eq: approximation of inverse I-F eta}
\end{equation}
\end{lemma}
\begin{proof}
See Appendix~\ref{app: coefficient matrix F eta}.
\end{proof}
\noindent As we shall see in Theorem~\ref{lemma: steady-state MSD performance}, it turns out that the decomposition in~\eqref{eq: relation inverse I-F} is very useful to highlight some important facts arising in the steady-state performance of the multitask algorithm.

\begin{lemma}{\emph{(Steady-state network performance)}}
\label{lemma: steady-state performance}Under Assumptions~\ref{assumption: strong convexity},~\ref{assumption: gradient noise},~\ref{assumption: combination matrix},~\ref{assumption: local smoothness hessian}, and~\ref{assumption: Smoothness condition on noise covariance}, and condition~\eqref{eq: condition on fourth-order moment of gradient noise}, it holds that:
\begin{align}
\limsup_{i\rightarrow\infty}\frac{1}{N}\expec\|\cw^o_{\eta}-\bcw_i\|^2&=\frac{1}{N}\sum_{n=0}^{\infty}\emph{\tr}\left(\cB_{\eta}^n\cY(\cB_{\eta}^\top)^n\right)+O(\mu^{1+{{\theta}}_m})\nonumber\\
&=\frac{1}{N}(\emph{\bvc}(\cY^\top))^{\top}(I-\cF_\eta)^{-1}\emph{\bvc}\left(I_{MN}\right)+O(\mu^{1+{{\theta}}_m}),\label{eq: steady-state network performance}
\end{align}
where ${{\theta}}_{m}=\frac{1}{2}\min\{1,{{\theta}}\}$, $\cB_{\eta}$ and $\cF_\eta$ are defined in~\eqref{eq: matrix B_eta} and~\eqref{eq: matrix F_eta}, and
\begin{align}
\cY&\triangleq\mu^2(I_{MN}-\mu\eta\cL)\cS_\eta(I_{MN}-\mu\eta\cL),\label{eq: definition cY}\\
\cS_\eta&\triangleq\emph{\diag}\left\{R_{s,k,\eta}\right\}_{k=1}^N,\label{eq: definition cS}
\end{align}
\end{lemma}
\begin{proof}
{The proof is a direct extension of the arguments used to establish Theorem 11.2 in~\cite{sayed2014adaptation} for single-task diffusion adaptation. See Appendix~\ref{app: steady-state performance}.}
\end{proof}

\begin{theorem}{\emph{(Network MSD performance)}}
\label{lemma: steady-state MSD performance}Under the same conditions of Lemma~\ref{lemma: steady-state performance}, it holds from Lemma~\ref{lemm: coefficient matrix} that
the steady-state network \emph{MSD} defined in~\eqref{eq: network MSD performance} can be written as:
\begin{equation}
\label{eq: final expression for the network MSD alternative}
\emph{MSD}=\frac{\mu}{2N}\sum_{m=1}^N\emph{\tr}\left(\left(\sum_{k=1}^N[v_m]^2_kH_{k,\eta}+\eta\lambda_mI\right)^{-1}\left(\sum_{k=1}^N[v_m]^2_kR_{s,k,\eta}\right)\right)+\frac{O(\mu)}{(O(1)+O(\eta))}.
\end{equation}
\end{theorem}
\begin{proof}
See Appendix~\ref{app: steady-state MSD performance}.
\end{proof}
\noindent As the derivation in Appendix~\ref{app: steady-state MSD performance} reveals, the second term on the {RHS} of~\eqref{eq: final expression for the network MSD alternative}  results from the matrix $W$ in~\eqref{eq: relation inverse I-F} which is zero when $\{H_{mn}=0,m\neq n\}$. When the Hessian matrices are uniform across the agents:
\begin{equation}
H_{k,\eta}\equiv H_\eta, \qquad k=1,\dots,N,
\end{equation}
we have $H_{mn}=0$ for $m\neq n$. 
 In this case, the network MSD in~\eqref{eq: final expression for the network MSD alternative} {simplifies} to:
\begin{equation}
\label{eq: final expression for the network MSD uniform covariance}
\begin{split}
\text{MSD}=\frac{\mu}{2N}\sum_{m=1}^N{\tr}\left(\left(H_{\eta}+\eta\lambda_mI\right)^{-1}\left(\sum_{k=1}^N[v_m]^2_kR_{s,k,\eta}\right)\right).
\end{split}
\end{equation}
Moreover, in the non-uniform Hessian matrices scenario, by letting $\eta=\mu^{-\epsilon}$ with $\epsilon>0$ chosen such that Assumption~\ref{assumption: combination matrix} is satisfied, we obtain:
\begin{equation}
\frac{O(\mu)}{(O(1)+O(\eta))}=O(\mu^{1+\epsilon}).
\end{equation}
In this case, the first term on the {RHS} of~\eqref{eq: final expression for the network MSD alternative} dominates the factor $O(\mu^{1+\epsilon})$ and when we evaluate the network MSD according to definition~\eqref{eq: network MSD performance}, the last term on the {RHS} of~\eqref{eq: final expression for the network MSD alternative} disappears when computing the limit as $\mu\rightarrow 0$. As we will see by simulations, the first  term on the {RHS} of~\eqref{eq: final expression for the network MSD alternative} provides a good approximation for the network MSD for any $ \eta\geq 0$.

{The first  term on the {RHS} of~\eqref{eq: final expression for the network MSD alternative} reveals explicitly the influence of the step-size $\mu$, regularization strength $\eta$, network topology (through the eigenvalues $\lambda_m$ and eigenvectors $v_m$ of the Laplacian),  gradient noise (through the covariance matrices $R_{s,k,\eta}$), and data characteristics (through the Hessian matrices $H_{k,\eta}$) on the network MSD performance. Observe that this term consists of the sum of $N$ individual terms, each associated with an eigenvalue $\lambda_m$ of the Laplacian matrix, and given by:
\begin{equation}
\label{eq: individual terms of the MSD}\frac{\mu}{2N}{\tr}\left(\left(\overline{H}_{m,\eta}+\eta\lambda_mI\right)^{-1}\overline{R}_{m,\eta}\right),
\end{equation}
where $\overline{H}_{m,\eta}$ and $\overline{R}_{m,\eta}$ are transformed versions of $\cH_{\eta}$ in~\eqref{eq: H of eta hessian} and $\cS_\eta$ in~\eqref{eq: definition cS} at the $m$-th eigenvalue $\lambda_m$, respectively:
\begin{align}
\overline{H}_{m,\eta}&\triangleq\sum_{k=1}^N[v_m]^2_kH_{k,\eta}=(v_m^\top\otimes I_M)\cH_{\eta}(v_m\otimes I_M),\\
\overline{R}_{m,\eta}&\triangleq\sum_{k=1}^N[v_m]^2_kR_{s,k,\eta}=(v_m^\top\otimes I_M)\cS_{\eta}(v_m\otimes I_M).
\end{align}
As shown in Section~\ref{subsec: uniform data profile}, under some assumptions on the data and noise characteristics, the individual terms in~\eqref{eq: individual terms of the MSD} are decaying functions of $\eta$ or $\lambda_m$.
}

{Before proceeding, we note that expression~\eqref{eq: final expression for the network MSD alternative} can be written alternatively as:
\begin{equation}
\label{eq: final expression for the network MSD}
\begin{split}
\text{MSD}=&\frac{\mu}{2N}{\tr}\left(\left(\sum_{k=1}^NH_{k,\eta}\right)^{-1}\left(\sum_{k=1}^NR_{s,k,\eta}\right)\right)+\\
&\frac{\mu}{2N}\sum_{m=2}^N{\tr}\left(\left(\sum_{k=1}^N[v_m]^2_kH_{k,\eta}+\eta\lambda_mI\right)^{-1}\left(\sum_{k=1}^N[v_m]^2_kR_{s,k,\eta}\right)\right)+\frac{O(\mu)}{(O(1)+O(\eta))},
\end{split}
\end{equation}
where we used the fact that $\lambda_1=0$ and $v_1=\frac{1}{\sqrt{N}}\mathds{1}_N$. Expression~\eqref{eq: final expression for the network MSD}} allows us to recover the network MSD performance of the single-task diffusion adaptation employed to {estimate $w^\star$ given by:
\begin{equation}
\label{eq: w star}
w^\star\triangleq\arg\min_{w}\sum_{k=1}^NJ_k(w).
\end{equation}}
To see this, we recall {the expression for the} network MSD performance of single-task diffusion adaptation derived in~\cite[pp.~594]{sayed2014adaptation}:
\begin{equation}
\label{eq: final expression for the network MSD diffusion}
\text{MSD}=\frac{\mu}{2N}\tr\left(\left(\sum_{k=1}^NH_{k,\star}\right)^{-1}\left(\sum_{k=1}^NR_{s,k,\star}\right)\right),
\end{equation}
where $H_{k,\star}\triangleq\nabla_{w_k}^2J_k(w^\star)$ and $R_{s,k,\star}$ is the covariance of the gradient noise in~\eqref{eq: R sk definition} at $w^\star$. In order to estimate $w^\star$  using the multitask strategy~\eqref{eq: distributed algorithm}, a very large $\eta$ needs to be chosen. In this case, we have $H_{k,\eta}=H_{k,\star}$ and $R_{s,k,\eta}=R_{s,k,\star}$. Moreover, the second and third terms on the {RHS} of expression~\eqref{eq: final expression for the network MSD} will be $O(\mu/\eta)$ which are negligible for a very large $\eta$. Thus, we obtain~\eqref{eq: final expression for the network MSD diffusion}.

\section{Multitask learning benefit}
\label{sec: multitask learning benefit}
Now that we have studied in {detail} the mean-square performance of the multitask strategy~\eqref{eq: distributed algorithm} relative to $\cw^o_{\eta}$, the minimizer of the regularized cost~\eqref{eq: global problem}, we will use the results to examine the benefit of multitask learning {compared to the non-cooperative solution} under the smoothness assumption.

Since each cost is strongly convex, each agent $k$ is able to estimate $w^o_k$ on its {own, if desired,} in a non-cooperative manner by running strategy~\eqref{eq: distributed algorithm} with $\eta=0$. {We know} from previous established results on single-agent adaptation~\cite[pp.~390]{sayed2014adaptation} that the network MSD {in that case will be given by:}
\begin{equation}
\label{eq: noncooperative MSD}
\text{MSD}_{\text{ncop}}=\frac{\mu}{2N}\tr\left(\sum_{k=1}^NH_{k,o}^{-1}R_{s,k,o}\right),
\end{equation}
where $H_{k,o}\triangleq \nabla_{w_k}^2J_k(w^o_k)$ and $R_{s,k,o}$ is the covariance of the gradient noise in~\eqref{eq: R sk definition} at $w^o_k$. {Note that expression~\eqref{eq: final expression for the network MSD alternative} allows us to recover the mean-square-error performance of stand-alone adaptive agents. In particular, it can be easily verified that~\eqref{eq: final expression for the network MSD alternative} reduces to expression $\frac{\mu}{2}\tr(H_{k,o}^{-1}R_{s,k,o})$ for the mean-square deviation of single agent learner~\cite[pp.~390]{sayed2014adaptation} when the network size is set to $N=1$ and the topology is removed. }

When $\eta>0$ is used, the graph Laplacian regularizer~\eqref{eq: global problem} induces a bias relative to $\cw^o$. However, when $\cw^o$ is smooth, we expect that promoting the smoothness of $\cw^o$ through regularization can improve the network MSD performance despite the bias induced in the estimation. 

\subsection{Induced bias relative to $\cw^o$}
The bias of the strategy~\eqref{eq: distributed algorithm} is given by:
\begin{equation}
\label{eq: steady state bias with respect to wo}
\expec\left(\cw^o-\bcw_i\right)=(\cw^o-\cw^o_{\eta})+\expec\left(\cw^o_{\eta}-\bcw_i\right).
\end{equation}
From the triangle inequality we have:
\begin{equation}
\limsup_{i\rightarrow\infty}\|\expec\left(\cw^o-\bcw_i\right)\|\leq\|\cw^o-\cw^o_{\eta}\|+\limsup_{i\rightarrow\infty}\|\expec\left(\cw^o_{\eta}-\bcw_i\right)\|,
\end{equation}
where $\limsup_{i\rightarrow\infty}\|\expec\left(\cw^o_{\eta}-\bcw_i\right)\|$ can be replaced by $\|\cwt'_{\infty}\|$ in~\eqref{eq: steady-state bias of long model}. {From expression~(31) in Part I~\cite{nassif2018diffusion} we know that the smoother $\cw^o$ is, the smaller $\|\cw^o-\cw^o_{\eta}\|$ will be. }
Furthermore, we know from {Theorem~4 in Part I~\cite{nassif2018diffusion}} that $\limsup_{i\rightarrow\infty}\|\expec\left(\cw^o_{\eta}-\bcw_i\right)\|$ is $O(\mu)$. Thus, for a smooth signal $\cw^o$ and a small $\mu$, the bias in~\eqref{eq: steady state bias with respect to wo} will be small. 

\subsection{Network MSD relative to $\cw^o$}
The mean-square-error performance of the strategy~\eqref{eq: distributed algorithm} relative to  $\cw^o$ is given by:
\begin{equation}
\label{eq: steady state MSE with respect to wo}
\expec\left\|\cw^o-\bcw_i\right\|^2=\|\cw^o-\cw^o_{\eta}\|^2+\expec\left\|\cw^o_{\eta}-\bcw_i\right\|^2+2(\cw^o-\cw^o_{\eta})^\top\expec\left(\cw^o_{\eta}-\bcw_i\right),
\end{equation}
and the network MSD can be expressed as:
\begin{equation}
\label{eq: steady state MSD with respect to wo}
\overline{\text{MSD}}=\frac{1}{N}\limsup_{i\rightarrow\infty}\expec\left\|\cw^o-\bcw_i\right\|^2=\text{MSD}+\frac{1}{N}\|\cw^o-\cw^o_{\eta}\|^2+\frac{2}{N}(\cw^o-\cw^o_{\eta})^\top\limsup_{i\rightarrow\infty}\expec\left(\cw^o_{\eta}-\bcw_i\right),
\end{equation}
where we used the bar notation to denote the network MSD relative to $\cw^o$ and where $\text{MSD}$ is given in~\eqref{eq: steady-state network performance} or~\eqref{eq: final expression for the network MSD alternative}. Note that  $\limsup_{i\rightarrow\infty}\expec\left(\cw^o_{\eta}-\bcw_i\right)$ can be replaced by $\cwt'_{\infty}$ in~\eqref{eq: steady-state bias of long model}. In order to improve the network MSD compared to the non-cooperative case~\eqref{eq: noncooperative MSD}, the regularization strength $\eta$ must be chosen such that:
\begin{equation}
\overline{\text{MSD}}\leq\text{MSD}_{\text{ncop}},
\end{equation}
and the optimal choice of $\eta$ is the one minimizing $\overline{\text{MSD}}$ in~\eqref{eq: steady state MSD with respect to wo} subject to $\eta\geq 0$.

\begin{figure}
\centering
\includegraphics[scale=0.535]{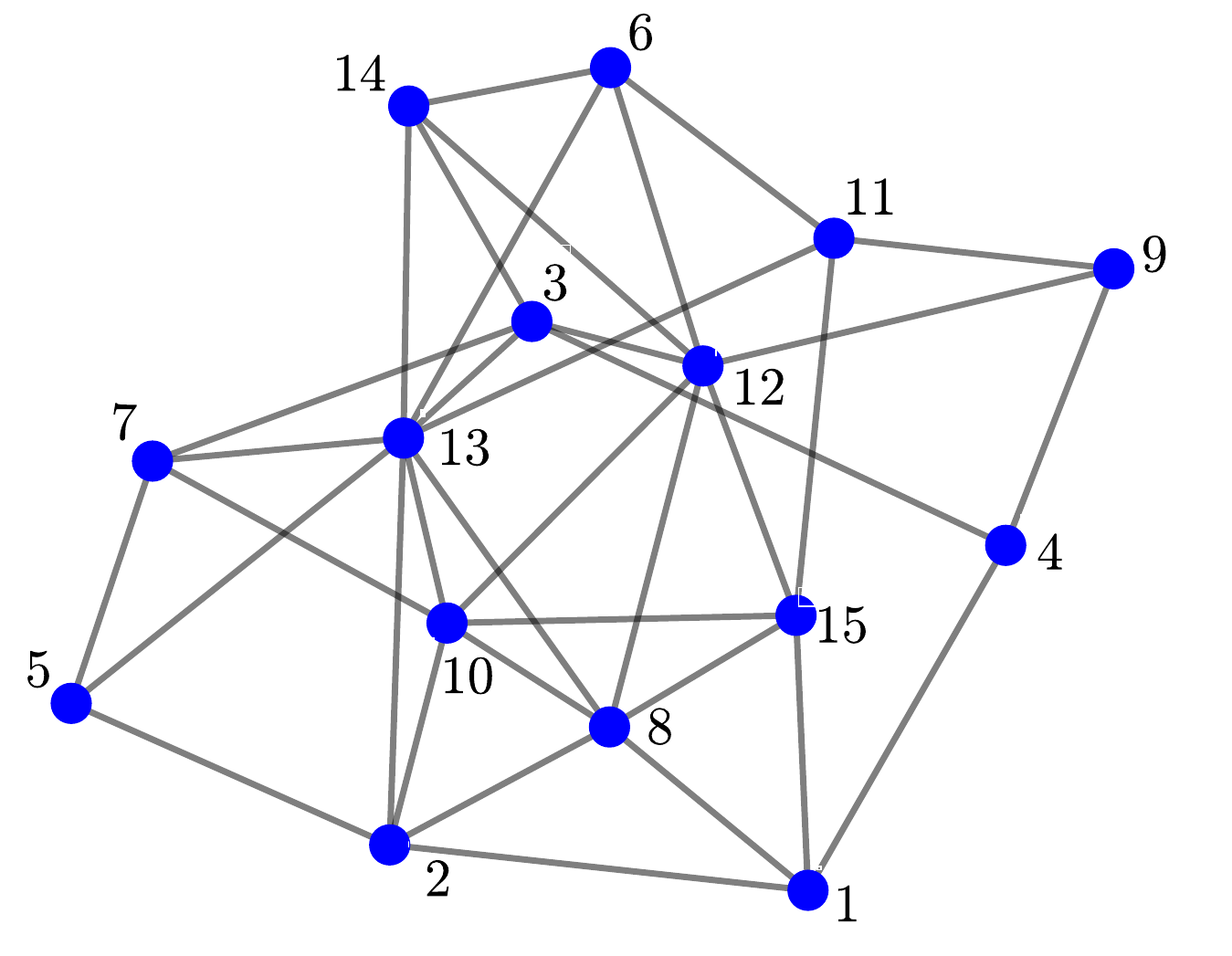}
\includegraphics[scale=0.435]{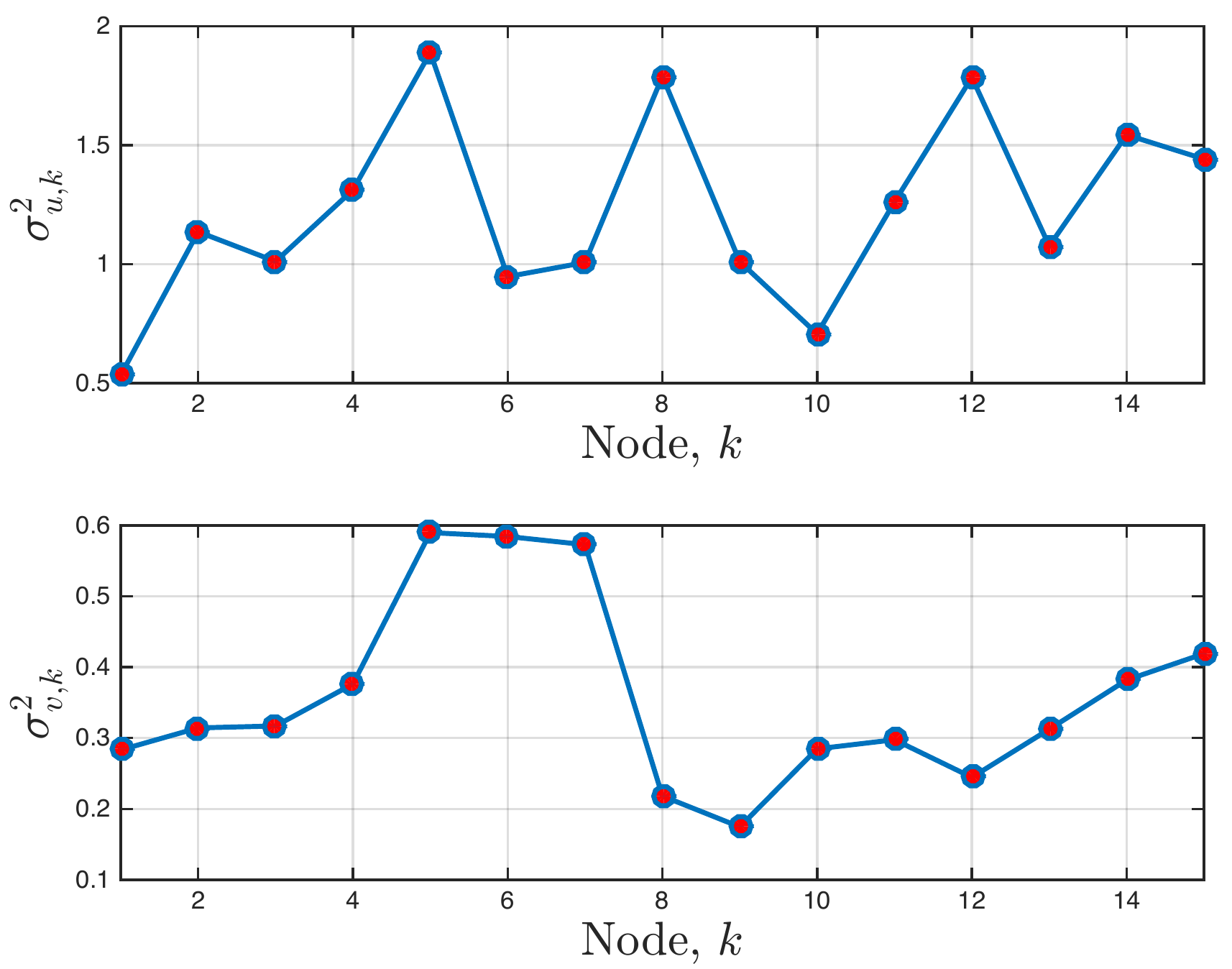}
\caption{Illustrative example. \textit{(Left)} Network topology. \textit{(Right)} Data and noise variances.}
\label{fig: network topology}
\end{figure}

\section{Simulation results}
\label{sex: experiments}
We consider a connected network of $N=15$ nodes and $M=5$ with the topology shown in Fig.~\ref{fig: network topology} (left). Each agent is subjected to streaming data $\{\boldsymbol{d}_k(i),\bu_{k,i}\}$ that are assumed to satisfy a linear regression model~\cite{sayed2014adaptation,chen2014multitask}:
\begin{equation}
\label{eq: linear data model}
\boldsymbol{d}_k(i)=\bu_{k,i} w^o_k+\bv_k(i),\qquad k=1,\ldots,N,
\end{equation}
for some unknown $M\times 1$ vector $w^o_k$ with $\bv_k(i)$ a measurement noise. A mean-square-error cost is associated with agent $k$:
\begin{equation}
\label{eq: MSE costs}
J_k(w_k)=\frac{1}{2}\,\expec|\boldsymbol{d}_k(i)-\bu_{k,i}w_k|^2, \qquad k=1,\ldots,N.
\end{equation}
The processes $\{\boldsymbol{d}_k(i),\bu_{k,i},\bv_k(i)\}$ are assumed to represent zero-mean jointly wide-sense stationary random processes satisfying: i) $\expec\bu^{\top}_{k,i}\bu_{\ell,j}=R_{u,k}\delta_{k,\ell}\delta_{i,j}$ where $R_{u,k}>0$ and the Kronecker delta  $\delta_{m,n}=1$ if $m=n$ and zero otherwise; ii) $\expec\bv_k(i)\bv_{\ell}(j)=\sigma^2_{v,k}\delta_{k,\ell}\delta_{i,j}$; iii) the regression and noise processes $\{\bu_{\ell,j},\bv_{k}(i)\}$ are independent of each other. We set $a_{k\ell}=0.1$ if $\ell\in\cN_k$ and $0$ otherwise. It turns out that the Laplacian matrix has 15 distinct eigenvalues. We generate  $\cw^o$ according to $\cw^o=\cV\cwb^o=\col\{\wb^o_m\}_{m=1}^N$ with:
\begin{equation}
\label{eq: spectral content of w o experiment}
\wb^o_m=\frac{1}{\sqrt{M}}\cdot\col\left\{e^{-\tau_j\lambda_m}\right\}_{j=1}^M.
\end{equation}
{Recall from Part I~\cite{nassif2018diffusion} that $S(\cw)$ in~\eqref{eq: quadratic regularization} can be written as:
\begin{equation}
\label{eq: S(w) smooth}
S(\cw)=\cw^\top\cL \cw=\sum_{m=2}^N\lambda_m\|{\wb_{m}}\|^2,
\end{equation}
where ${\wb_{m}}=(v_m^\top\otimes I_M)\cw$.} From~\eqref{eq: S(w) smooth}, we observe that the larger $\{\tau_j\geq 0\}$ are, the smoother the signal $\cw^o$ is. Note that, for MSE networks, it holds that $H_k(w_k)=\nabla_{w_k}^2J_k(w_k)=R_{u,k}$ $\forall w_k$. Thus, the fixed point bias $\cwt_{\infty}=\cw_\eta^o-\cw_{\infty}$ {given by~(49) in Part I~\cite{nassif2018diffusion}} is equal to the long-term model bias $\cwt_{\infty}'$ in~\eqref{eq: steady-state bias of long model}. Furthermore, from the definition~\eqref{eq: gradient noise process}, the gradient noise process at agent $k$ is given by:
\begin{equation}
\label{eq: MSE network gradient noise}
\bs_{k,i}(\bw_{k})=(\bu_{k,i}^{\top}\bu_{k,i}-R_{u,k})(w^o_k-\bw_{k})+\bu_{k,i}^\top\bv_{k}(i).
\end{equation}
Consider the covariance $R_{s,k,\eta}$ defined in~\eqref{eq: R sk definition}. From~\eqref{eq: MSE network gradient noise}, we have:
\begin{equation}
\label{eq: covariance noise MSE}
R_{s,k,\eta}=\expec[(\bu_{k,i}^{\top}\bu_{k,i}-R_{u,k})W_{k,\eta}(\bu_{k,i}^{\top}\bu_{k,i}-R_{u,k})]+\sigma^2_{v,k}R_{u,k},
\end{equation}
where
\begin{equation}
\label{eq: W k eta}
W_{k,\eta}=(w^o_k-w^o_{k,\eta})(w^o_k-w^o_{k,\eta})^\top.
\end{equation}
To evaluate~\eqref{eq: covariance noise MSE}, we need the fourth order moment of the regressors. Let us assume that the regressors are zero-mean real Gaussian. In this case, using the fact that $W_{k,\eta}$ is symmetric, we obtain~\cite{isserlis1918formula}:
\begin{equation}
\expec[\bu_{k,i}^{\top}\bu_{k,i}W_{k,\eta}\bu_{k,i}^{\top}\bu_{k,i}]=2R_{u,k}W_{k,\eta}R_{u,k}+R_{u,k}\tr(R_{u,k}W_{k,\eta}).
\end{equation}
Replacing the above expression in~\eqref{eq: covariance noise MSE}, we obtain:
\begin{equation}
\label{eq: covariance noise MSE 1}
R_{s,k,\eta}=R_{u,k}W_{k,\eta}R_{u,k}+R_{u,k}\tr(R_{u,k}W_{k,\eta})+\sigma^2_{v,k}R_{u,k}.
\end{equation}


\subsection{Uniform data profile}
\label{subsec: uniform data profile}
In this setting, we assume uniform covariance matrices scenario where $R_{u,k}=R_u$ $\forall k$. This scenario is encountered for example in distributed denoising problems in wireless sensor networks (or image denoising)~\cite{shuman2013emerging,shuman2018distributed}. In this case, $R_u$ is equal to one. In such problems, the $N$ sensors in the network are observing $N$-dimensional signal, with each entry of the signal corresponding to one sensor. Using the prior knowledge that the signal is smooth w.r.t. the underlying topology, the sensor task is to denoise the corresponding entry of the signal by performing local computations and cooperating with its neighbor in order to improve the error variance.

It turns out that in this scenario, the output of the network $\cw^o_{\eta}$ can be interpreted as the output of a low-pass graph filter applied to the graph signal $\cw^o$~\cite{nassif2018distributed,shuman2013emerging,sandryhaila2013discrete}. To see this, let us first recall the notion of graph frequencies, graph Fourier transform, and graph {filters~\cite{shuman2013emerging,sandryhaila2013discrete,ortega2018graph,shuman2018distributed}}. Consider the connected graph $\cG = \{\cN, \cE,A\}$ equipped with a Laplacian matrix $L$, which can be decomposed as $L =
V\Lambda V^{\top}$. A graph signal supported on the set $\cN$ is defined as a vector $x\in\mathbb{R}^N$ whose $k$-th component $x_k\in\mathbb{R}$ represents the value of the signal at the $k$-th node. By analogy to the classical Fourier analysis, the eigenvectors of the Laplacian matrix $L$ are used to define a graph Fourier basis $V$ and the eigenvalues are considered as the graph frequencies. The graph Fourier transform (GFT) transforms a graph signal $x$ into the graph frequency domain according to $\overline{x}=V^{\top}x$ where $\{\overline{x}_1,\ldots,\overline{x}_N\}$ are called the spectrum of $x$. A graph filter $\Phi$ is an operator that acts upon a graph signal $x$ by amplifying or attenuating its spectrum as: $\Phi=\sum_{m=1}^N\Phi(\lambda_m)\overline{x}_mv_m$. The frequency response of the filter $\Phi(\lambda)$ controls how much $\Phi$ amplifies the signal spectrum. Low frequencies correspond to small eigenvalues, and low-pass or smooth filters correspond to decaying functions $\Phi(\lambda)$. Since we are dealing with vectors $w_k\in\mathbb{R}^M$ instead of scalars $x_k$, the graph transformation $\overline{x}=V^{\top}x$ becomes $\cwb = (V^{\top}\otimes I_M)\cw$. In the uniform covariance matrices scenario, we have $\cH^o_{\eta}=I_N\otimes R_u$ and {relation~(24) in Part I~\cite{nassif2018diffusion}} reduces to:
\begin{equation}
\label{eq: output of the filter for uniform covariances}
{\wb^o_{m,\eta}}=\left(I_M-\eta\lambda_m\left(\eta\lambda_mI_M+R_u\right)^{-1}\right)\wb^o_{m}, \quad m=1,\ldots,N.
\end{equation}
since $(v_m^\top\otimes I_M)\cH^o_{\eta}(v_n\otimes I_M)=R_u$ if $m=n$ and zero otherwise. {By applying the matrix inversion identity~\cite{kailath1980linear},} it holds that the $m$-th subvector corresponding to the $m$-th eigenvalue (or graph frequency) of $\cwb^o_{\eta}=(V^{\top}\otimes I_M)\cw^o$ satisfies:
\begin{align}
{\wb^o_{m,\eta}}&=\left(\eta\lambda_mI_M+R_u\right)^{-1}R_u \wb^o_{m}, \label{eq: output of the filter for uniform covariances mse}\\
\|{\wb^o_{m,\eta}}\|&~\leq\frac{1}{1+\eta\frac{\lambda_m}{\lambda_{\max}\left(R_u\right)}}\|\wb^o_{m}\|,\label{eq: output of the filter for uniform covariances norm}
\end{align}
If $\eta=0$, we are in the case of an all-pass graph filter since the frequency content of the output signal $\cw^o_{\eta}$, is the same as the frequency content of the input signal $\cw^o$. For $\eta>0$, we are in the case of a low-pass graph filter since the norm of the $m$-th frequency content of the output signal $\cw^o_{\eta}$, namely, $\|{\wb^o_{m,\eta}}\|$, is less than or equal to the norm of the $m$-th frequency content of the input signal  $\cw^o$, namely, $\|\wb^o_{m}\|$. For fixed $\eta$, as $m$ increases, the ratio in~\eqref{eq: output of the filter for uniform covariances norm} decreases. This validates the low-pass filter interpretation. The regularization parameter $\eta$ controls the sharpness of the low-pass filter. We illustrate this in Fig.~\ref{fig: uniform data profile} (left). In the experiment, we set $R_u=I_M$ and $\tau_j=0$ in~\eqref{eq: spectral content of w o experiment} so that $\|\wb^o_m\|^2=1$ $\forall m$ and we illustrate in Fig.~\ref{fig: uniform data profile} (left) the squared $\ell_2$-norm of ${\wb^o_{m,\eta}}$ for different values of $\eta$ from~\eqref{eq: output of the filter for uniform covariances mse}. In order to visualize the frequency response of the graph filter, we also plot in dashed lines the ratio $\frac{\|\wb^o_{m,\eta}\|^2}{\|\wb^o_m\|^2}$ from~\eqref{eq: output of the filter for uniform covariances mse} for $\lambda_m\in[0,1.2]$ with $\|\wb^o_m\|^2=1$ $\forall m$.

We observe that a similar behavior arises when studying the network MSD for smooth signal $\cw^o$. When $\cw^o$ is smooth, $W_{k,\eta}$ in~\eqref{eq: W k eta} is small and in this case, the covariance matrix $R_{s,k,\eta}$ in~\eqref{eq: covariance noise MSE 1} can be approximated by $R_{s,k,\eta}\approx \sigma^2_{v,k}R_{u}$. In this case, the network MSD expression in~\eqref{eq: final expression for the network MSD uniform covariance} can be approximated as:
\begin{equation}
\label{eq: final expression for the network MSD case 1 a}
\begin{split}
\text{MSD}\approx\frac{\mu M}{2}\frac{1}{N^2}\left(\sum_{k=1}^N\sigma^2_{v,k}\right)+\frac{\mu}{2N}\sum_{m=2}^N\left(\sum_{k=1}^N[v_m]^2_k\sigma^2_{v,k}\right){\tr}\left(\left(I_M+\eta\lambda_mR_u^{-1}\right)^{-1}\right).
\end{split}
\end{equation}
Since the trace of a matrix is equal to the sum of its eigenvalues, we have:
\begin{equation}
\label{eq: trace relation}
{\tr}\left(\left(I_M+\eta\lambda_mR_u^{-1}\right)^{-1}\right)=\sum_{q=1}^M\frac{1}{1+\frac{\eta\lambda_m}{\lambda_{q}(R_u)}}.
\end{equation}
The above expression shows that for a fixed $\lambda_m$, as $\eta$ increases, the above trace decreases. We conclude that, when $\eta$ increases, the sum on the {RHS} of~\eqref{eq: final expression for the network MSD case 1 a} decreases. By further assuming uniform noise profile, i.e., $\sigma^2_{v,k}=\sigma^2_{v}$ $\forall k$, expression~\eqref{eq: final expression for the network MSD case 1 a} reduces to:
\begin{equation}
\label{eq: final expression for the network MSD case 1 b}
\text{MSD}=\sum_{m=1}^N\text{MSD}(\lambda_m)\qquad\text{with}\qquad\text{MSD}(\lambda_m)\approx\frac{\mu}{2N}\sigma^2_{v}{\tr}\left(\left(I_M+\eta\lambda_mR_u^{-1}\right)^{-1}\right).
\end{equation}
From~\eqref{eq: trace relation}, we conclude that, for a fixed $\eta$, as $\lambda_m$ increases, the corresponding trace term in~\eqref{eq: final expression for the network MSD case 1 b} decreases. This case is illustrated numerically in Fig.~\ref{fig: uniform data profile} (right).  In the experiment, we set $R_u=I_M$, $\sigma^2_v=0.1$,  $\mu=0.005$, and $\tau_j=7+j$ in~\eqref{eq: spectral content of w o experiment} so that $\cw^o$ is smooth and we illustrate $\text{MSD}(\lambda_m)$ in~\eqref{eq: final expression for the network MSD case 1 b} for different values of $\eta$ .
\begin{figure}
\centering
\includegraphics[scale=0.5]{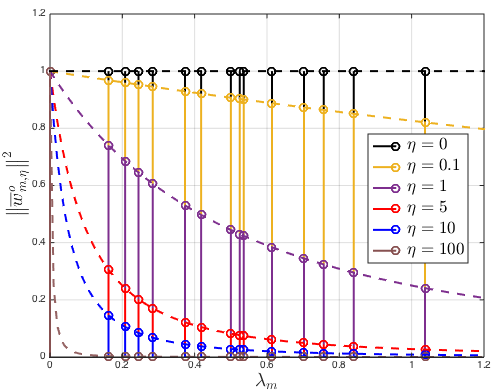}
\includegraphics[scale=0.5]{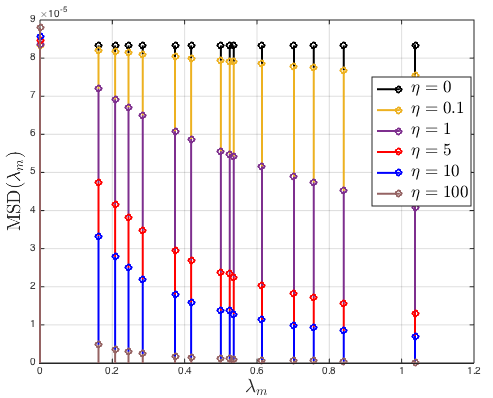}
\caption{Uniform data profile scenario. \textit{(Left)} Spectral content of $\cw^o_{\eta}$ from~\eqref{eq: output of the filter for uniform covariances mse} for different $\eta$.  Dashed lines correspond to the ratio $\frac{\|\wb^o_{m,\eta}\|^2}{\|\wb^o_m\|^2}$ for $\lambda_m\in[0,1.2]$. \textit{(Right)}  Network MSD at $\lambda_m$ from~\eqref{eq: final expression for the network MSD case 1 b} (relative to $\cw^o_{\eta}$ with $\cw^o$ a {smooth} signal) for different $\eta$  at $\mu=0.005$.}
\label{fig: uniform data profile}
\end{figure}
\subsection{Varying data profile}

We assume that $R_{u,k}=\sigma^2_{u,k}I_M$ for all $k$. In this case, expression~\eqref{eq: covariance noise MSE 1} reduces to:
\begin{equation}
R_{s,k,\eta}=\sigma^4_{u,k}(W_{k,\eta}+\|w^o_k-w^o_{k,\eta}\|^2I_M)+\sigma^2_{v,k}\sigma^2_{u,k}I_M.
\end{equation}
The variances $\sigma^2_{u,k}$ and $\sigma^2_{v,k}$ are given in Fig.~\ref{fig: network topology} (right). We set $\tau_j=j$ in~\eqref{eq: spectral content of w o experiment}.

In order to characterize the influence of the step-size $\mu$ and the regularization parameter $\eta$ on the performance of the algorithm relative to $\cw^o_{\eta}$, we report in Fig.~\ref{fig: network MSD} (left) the squared norm of the fixed point bias $\cwt_{\infty}=\cw_\eta^o-\cw_{\infty}$ {given by~(49) in Part I~\cite{nassif2018diffusion}} (which is equal to the long-term model bias $\cwt'_{\infty}$  in~\eqref{eq: steady-state bias of long model}) as a function of $\eta$ where $\eta=\mu^{-\epsilon}$ with $\epsilon\geq-1$ for $\mu =\{10^{-3}, 10^{-4}, 10^{-5}\}$. We observe that for small $\eta$, the squared norm of the bias increases $40$ dB per decade (when $\eta$ goes from $\eta_1$ to $10\eta_1$). This means that the squared norm of the bias is on the order of $\eta^4$ for fixed $\mu$. For large $\eta$, the bias becomes constant and we see that, when $\mu$ goes from $\mu_1$ to $10\mu_1$, it increases $20$ dB. This means that the squared norm of the bias is on the order of $\mu^2$. Finally, we observe that for fixed $\eta$, the squared norm of the bias is on the order of $\mu^2$. In Fig.~\ref{fig: network MSD} (middle), we report the network MSD learning curves relative to $\cw^o_{\eta}$ obtained by running strategy~\eqref{eq: distributed algorithm} for $\eta=5$ and for two different values of $\mu$. The curves are obtained by averaging the trajectories $\{\frac{1}{N}\expec\|\cw^o_{\eta}-\bcw_i\|^2\}$ over $200$ repeated experiments. In the simulations, we use the following approximation for the network MSD expression in~\eqref{eq: final expression for the network MSD alternative}:
\begin{equation}
\label{eq: final expression for the network MSD case 2}
\text{MSD}_{\text{app}}=
\frac{\mu}{2N}\sum_{m=1}^N\left(\sum_{k=1}^N[v_m]^2_k\sigma^2_{u,k}+\eta\lambda_m\right)^{-1}{\tr}\left(\sum_{k=1}^N[v_m]^2_kR_{s,k,\eta}\right),
\end{equation}
where we use the subscript ``app'' to indicate that it is an approximation. Compared with~\eqref{eq: final expression for the network MSD alternative}, we see that the term $\frac{O(\mu)}{(O(1)+O(\eta))}$ has been removed in~\eqref{eq: final expression for the network MSD case 2}. It is observed that the learning curves tend to the same MSD value predicted by the theoretical expression~\eqref{eq: steady-state network performance} (with $O(\mu^{1+{{\theta}}_m})$ removed). Furthermore, we observe that the MSD predicted by expression~\eqref{eq: final expression for the network MSD case 2} provides a good approximation for the performance of strategy~\eqref{eq: distributed algorithm}.  Finally, it can be observed that the MSD is on the order of $\mu$. In Fig.~\ref{fig: network topology} (right), we report the MSD learning curves relative to $\cw^o_{\eta}$ obtained by running strategy~\eqref{eq: distributed algorithm} for $\mu=10^{-3}$ and for six different values of $\eta$. As it can be seen, by increasing $\eta$, the network MSD decreases. Furthermore, it is observed that expression~\eqref{eq: final expression for the network MSD case 2} provides a good approximation for the network MSD for any $\eta\geq 0$.

In Fig.~\ref{fig: network MSD bar 1}, we characterize the influence of $\eta$ on the performance of strategy~\eqref{eq: distributed algorithm} relative to $\cw^o$ with $\cw^o$ a smooth signal generated according to~\eqref{eq: spectral content of w o experiment} with $\tau_j=7+j$. We set $\mu=5\cdot10^{-3}$. In Fig.~\ref{fig: network MSD bar 1} (left), we plot $\overline{\text{MSD}}$ for $\eta\in[0,350]$. To generate $\overline{\text{MSD}}$ we use~\eqref{eq: steady state MSD with respect to wo} with $\text{MSD}$ replaced by $\text{MSD}_{\text{app}}$ in~\eqref{eq: final expression for the network MSD case 2} which has a low computational complexity. As it can be seen from this plot, $\eta=4$ gives the best $\overline{\text{MSD}}$. In Fig.~\ref{fig: network MSD bar 1} (right), we report the network learning curves $\frac{1}{N}\expec\|\cw^o-\bcw_i\|^2$ for $\eta=\{0,4,350\}$.  It is observed that the learning curves tend to the same $\overline{\text{MSD}}$ value predicted by the theoretical expression~\eqref{eq: steady state MSD with respect to wo}  with $\text{MSD}$ replaced by~\eqref{eq: steady-state network performance} (with $O(\mu^{1+{{\theta}}_m})$ removed) or~\eqref{eq: final expression for the network MSD case 2}.

\begin{figure*}
\centering
\includegraphics[scale=0.334]{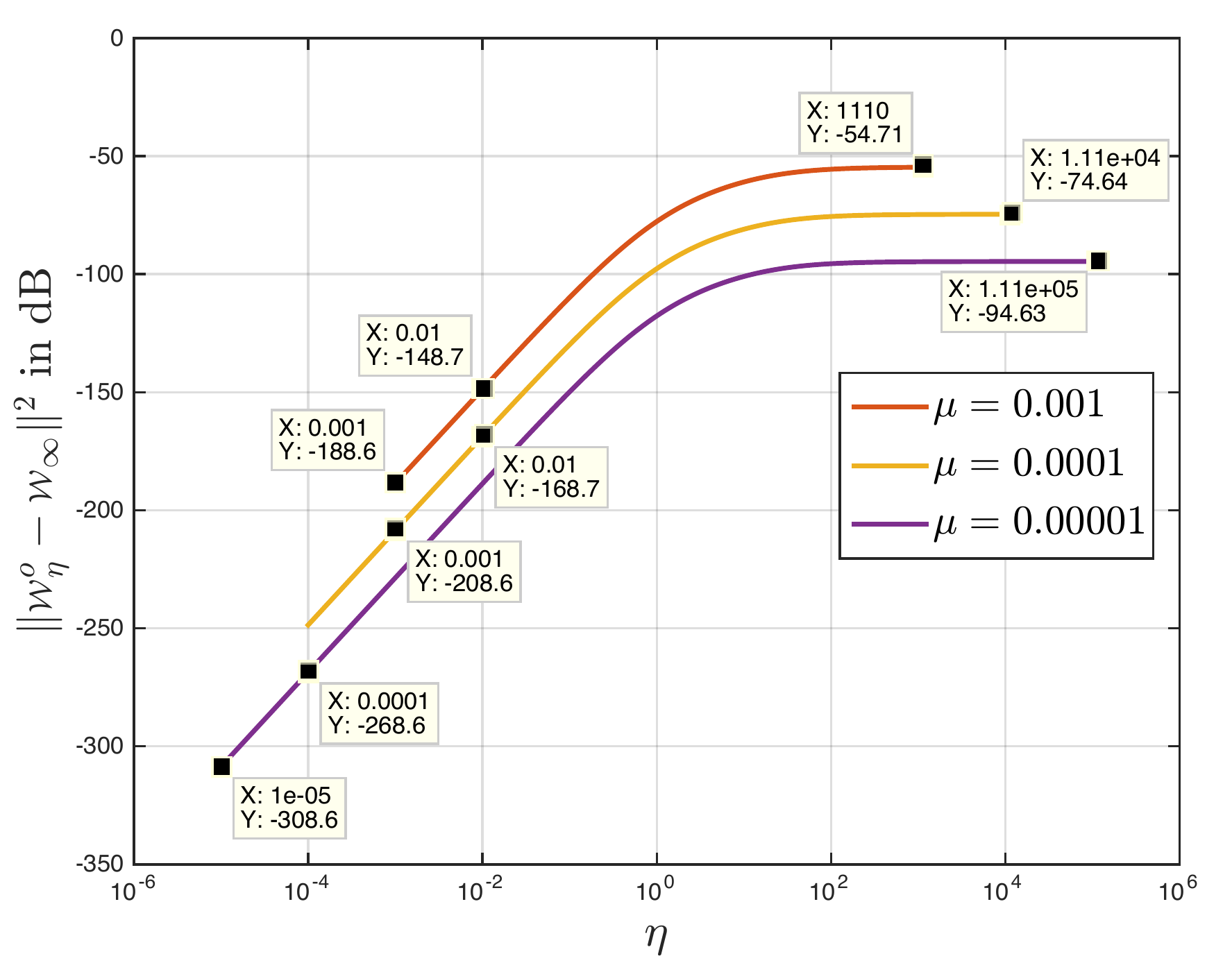}
\includegraphics[scale=0.334]{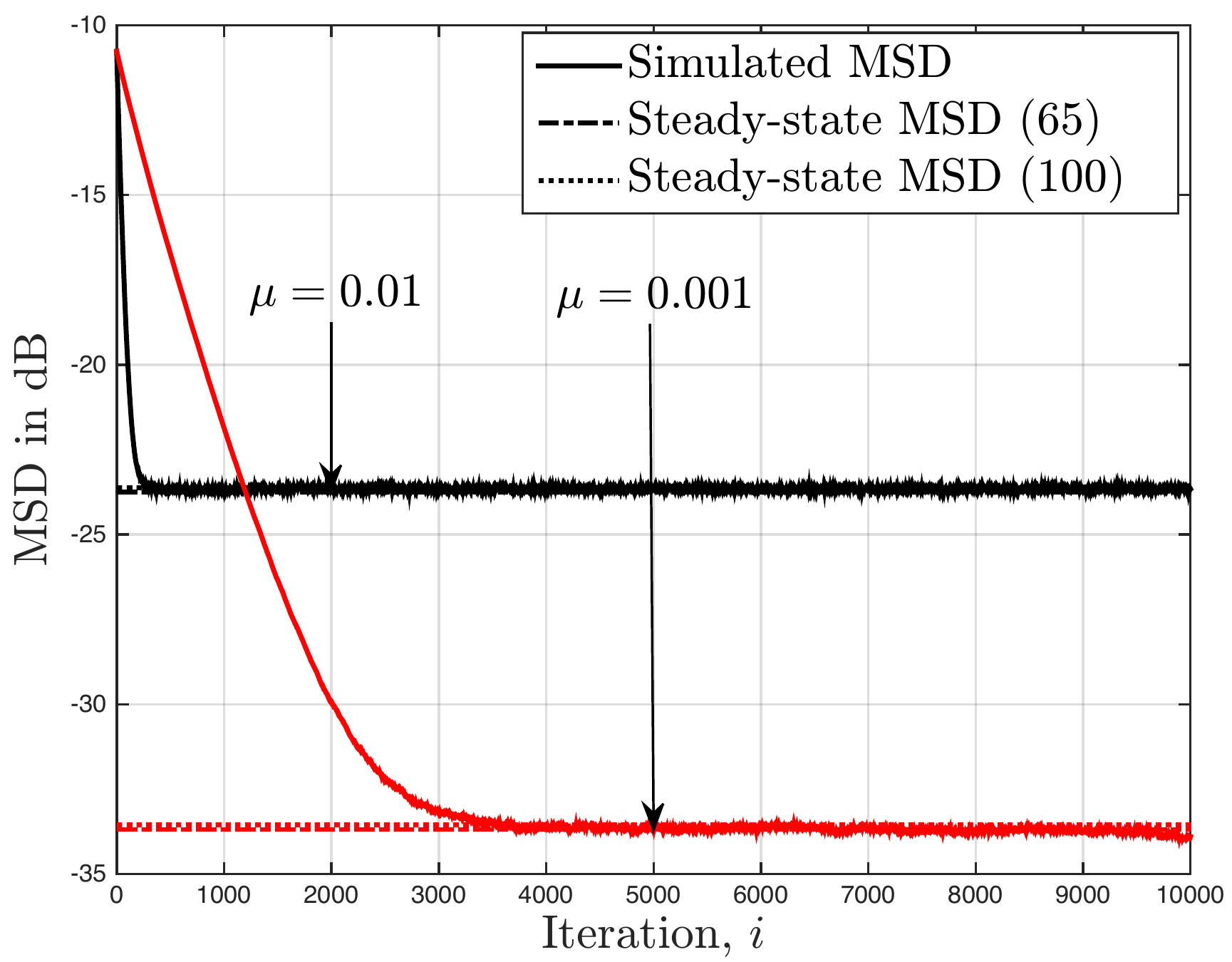}
\includegraphics[scale=0.334]{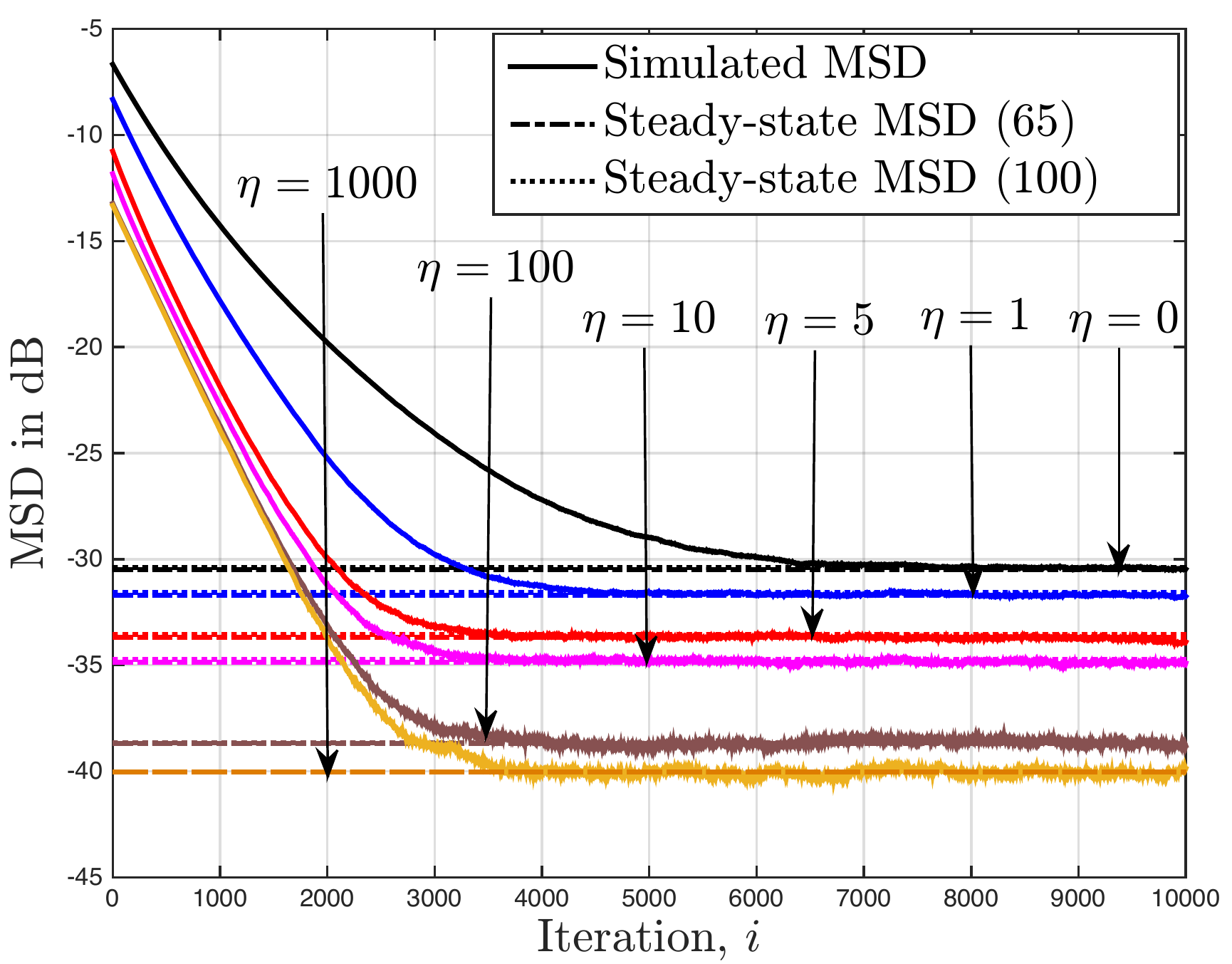}
\caption{Network performance relative to $\cw^o_{\eta}$.  \textit{(Left)} Squared $\ell_2$-norm of the {bias~\eqref{eq: steady-state bias of long model}}. \textit{(Middle)} Evolution of the learning curves for fixed regularization strength $\eta=5$, varying step-size $\mu$. \textit{(Right)} Evolution of the learning curves for fixed  $\mu=0.001$, varying $\eta$.}
\label{fig: network MSD}
\end{figure*}

\begin{figure*}
\centering
\includegraphics[scale=0.5]{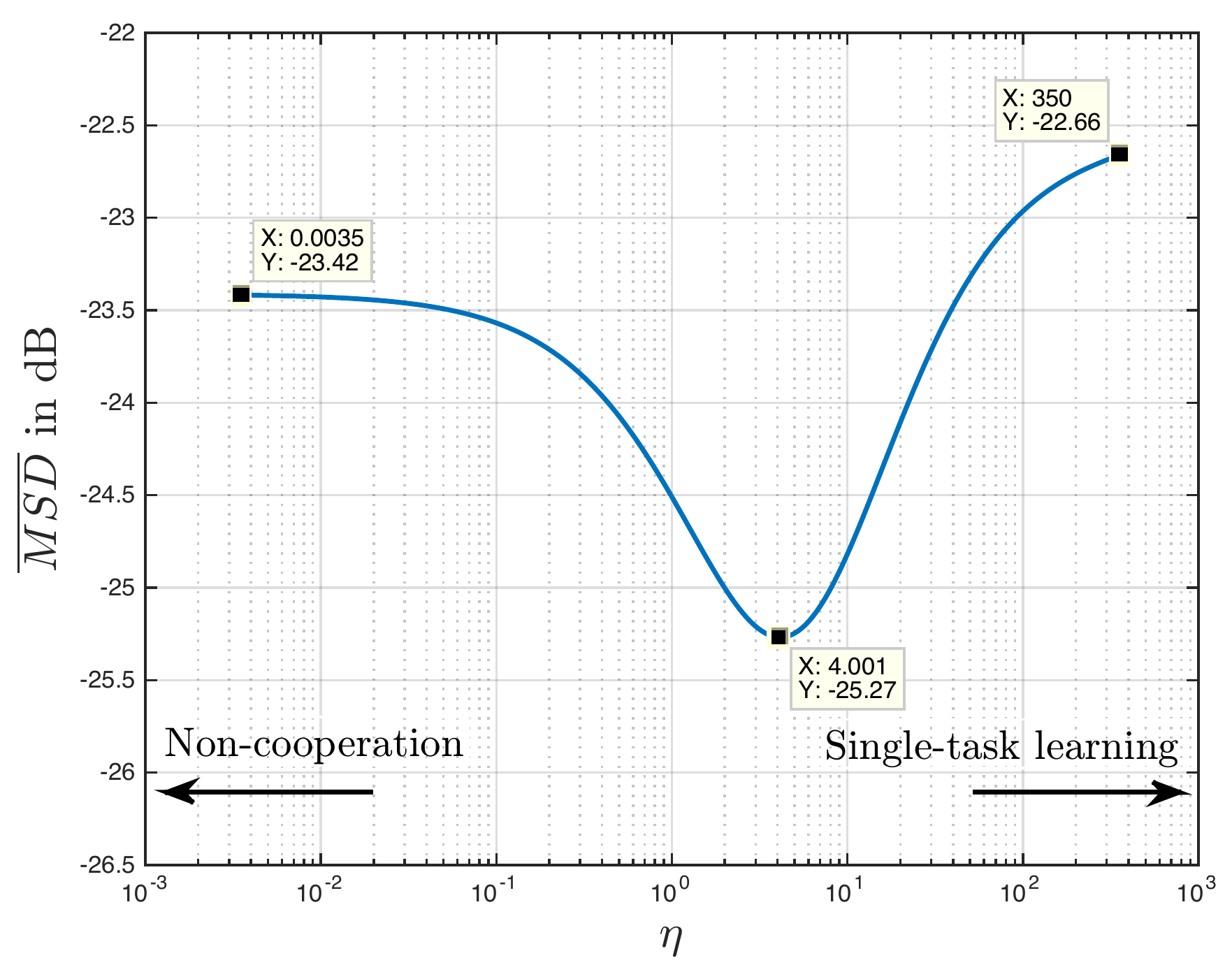}
\includegraphics[scale=0.5]{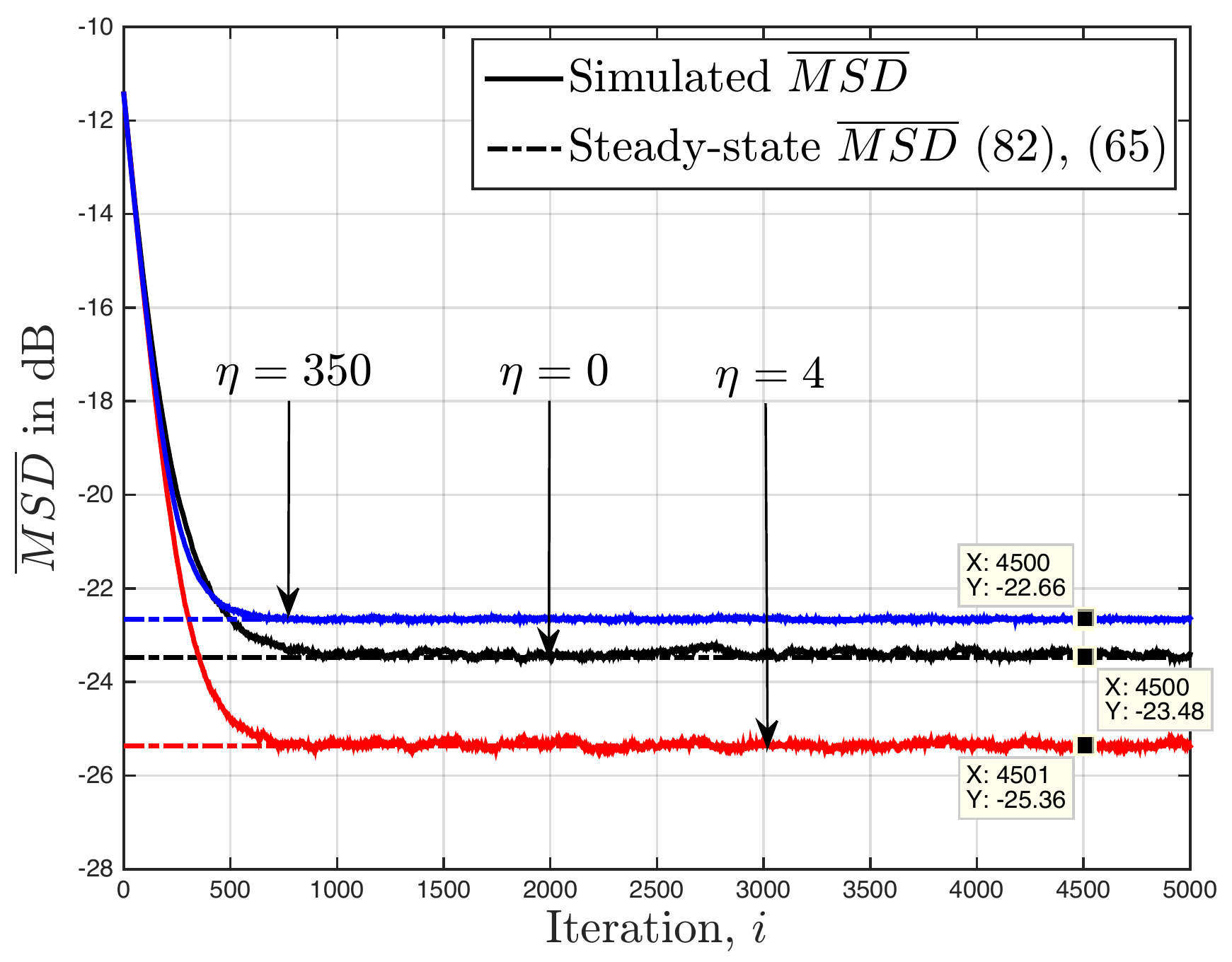}
\caption{Network  performance relative to $\cw^o$ with a smooth signal $\cw^o$ at $\mu=0.005$. \textit{(Left)} Network $\overline{\text{MSD}}$ as a function of the regularization strength $\eta\in[0,350]$. \textit{(Right)} Evolution of the network $\overline{\text{MSD}}$ learning curve for three different values of $\eta$.}
\label{fig: network MSD bar 1}
\end{figure*}

\section{Conclusion}
In this {paper, and its accompanying Part I~\cite{nassif2018diffusion},} we considered multitask inference problems where agents in the network have individual parameter vectors to estimate subject to a smoothness condition over the graph. Based on diffusion adaptation, we proposed a strategy that allows the network to minimize a global cost consisting of the aggregate sum of the individual costs regularized by a term promoting smoothness. We showed that, for small step-size parameter, the network is able to approach the minimizer of the regularized problem to arbitrarily good accuracy levels. Furthermore, we showed how the regularization strength influences the limit point and the steady-state mean-square-error (MSE) performance of the algorithm. Analytical expressions illustrating these effects are derived. These expressions revealed explicitly the influence of the network topology, data settings, step-size parameter, and regularization strength on the network MSE performance and provided insights into the design of effective multitask strategies for distributed inference over networks.  Illustrative examples were considered and links to spectral graph filtering were also provided.

\begin{appendices}

\section{Size of the perturbation sequence $\bc_{i-1}$}
\label{app: perturbation sequence size}
\noindent The argument is similar to the one presented in~\cite[Sec. 10.1]{sayed2014adaptation}. For ease of reference, we {provide a sketch of the proof here.} Let $\bwt_{k,i-1}=w^o_{k,\eta}-\bw_{k,i-1}$. For each agent $k$, we have {from eq.~(174) in Part I~\cite{nassif2018diffusion}}:
\begin{align}
\|\widetilde{\bH}_{k,i-1}\|&\triangleq\|H_{k,\eta}-\bH_{k,i-1}\|\nonumber\\
&\leq\int_{0}^1\|\nabla_{w_k}^2J_k(w_{k,\eta}^o)-\nabla_{w_k}^2J_k(w_{k,\eta}^o-t\bwt_{k,i-1})\|dt\nonumber\\
&\leq\int_{0}^1\kappa'_d\|t\bwt_{k,i-1}\|dt=\frac{1}{2}\kappa'_d\|\bwt_{k,i-1}\|,\label{eq: bound on norm of H tilde k}
\end{align}
{so that}
\begin{equation}
\|\widetilde{\bcH}_{i-1}\|=\max_{1\leq k\leq N}\|\widetilde{\bH}_{k,i-1}\|\leq\frac{1}{2}\kappa'_d\max_{1\leq k\leq N}\|\bwt_{k,i-1}\|\leq\frac{1}{2}\kappa'_d\|\bcwt_{i-1}\|,
\end{equation}
and, consequently,
\begin{align}
\limsup_{i\rightarrow\infty}\expec\|\bc_{i-1}\|&\leq\limsup_{i\rightarrow\infty}\expec\|\widetilde{\bcH}_{i-1}\|\|\bcwt_{i-1}\|\nonumber\\
&\leq\frac{1}{2}\kappa'_d\limsup_{i\rightarrow\infty}\expec\|\bcwt_{i-1}\|^2=O(\mu),\label{eq: bound on norm of c}
\end{align}
from {Theorem~2 in Part I~\cite{nassif2018diffusion}}.

In the following, we {argue} that $\|\bc_{i-1}\|=O(\mu)$ asymptotically with \textit{high probability}. Let us introduce the nonnegative random variable $\bu=\|\bc_{i-1}\|$ and let us recall Markov's inequality which states that for any \textit{nonnegative} random variable $\bu$ and $\xi>0$, it holds that:
\begin{equation}
\label{eq: Markov inequality}
\text{Prob}(\bu\geq\xi)\leq\left(\expec\bu\right)/\xi.
\end{equation}
Let $r_c=n\mu$, for any constant integer $n\geq 1$ that we are free to choose. We then conclude from~\eqref{eq: bound on norm of c} and~\eqref{eq: Markov inequality} that for $i\gg1$:
\begin{align}
\text{Prob}(\|\bc_{i-1}\|<r_c)&=1-\text{Prob}(\|\bc_{i-1}\|\geq r_c)\notag\\
&\geq1-\expec\|\bc_{i-1}\|/r_c\notag\\
&\geq1-O(1/n)
\end{align}
where the term $O(1/n)$ is independent of $\mu$. This result shows that the probability of having  $\|\bc_{i-1}\|$ bounded by $r_c$ can be made arbitrarily close to one by selecting a large enough value for $n$. Once the value for $n$ has been fixed to meet a desired confidence level, then $r_c=O(\mu)$.

\section{Proof of Lemma~\ref{lemma: dimension of the approximation error}}
\label{app: proof of dimension of the approximation error}
\noindent To simplify the notation, we introduce the difference:
\begin{equation}
\bz_i\triangleq\bcwt_i-\bcwt_i'.
\end{equation}
Subtracting recursions~\eqref{eq: error recursion wt 0}  and~\eqref{eq: long term error model 1}, we get:
\begin{equation}
\label{eq: recursion for z_i}
\bz_i=\cB_{\eta}\bz_{i-1}+\mu(I_{MN}-\mu\eta\cL)\bc_{i-1}.
\end{equation}
in terms of the random perturbation sequence $\bc_{i-1}$ given in~\eqref{eq: definition of c i-1}. For each agent $k$, we have from~\eqref{eq: bound on norm of H tilde k}:
\begin{equation}
\begin{split}
\|\widetilde{\bH}_{k,i-1}\|&\leq\frac{1}{2}\kappa'_d\|w^o_{k,\eta}-\bw_{k,i-1}\|.
\end{split}
\end{equation}
{It follows that}
\begin{equation}
\begin{split}
\|\bc_{k,i-1}\|&\leq\|\widetilde{\bH}_{k,i-1}\|\|w^o_{k,\eta}-\bw_{k,i-1}\|\\
&\leq\frac{1}{2}\kappa'_d\|w^o_{k,\eta}-\bw_{k,i-1}\|^2\\
&\leq\frac{1}{2}\kappa'_d\|\cw^o_{\eta}-\bcw_{i-1}\|^2.
\end{split}
\end{equation}
and
\begin{equation}
\label{eq: bound on squared ck}
\|\bc_{k,i-1}\|^2\leq\frac{1}{4}(\kappa'_d)^2\|\cw^o_{\eta}-\bcw_{i-1}\|^4.
\end{equation}
Applying Jensen's inequality~\cite[pp.~77]{boyd2004convex} {to} the convex function $\|\cdot\|^2$, we obtain from~\eqref{eq: recursion for z_i}:
\begin{equation}
\label{eq: jensen's inequality for the difference}
\expec\|\bz_{k,i}\|^2\leq\sum_{\ell\in\cN_k}[C]_{k,\ell}\expec\|(I-\mu H_{\ell,\eta})\bz_{\ell,i-1}+\mu\bc_{\ell,i-1}\|^2.
\end{equation}
where $C$ {is given by:
\begin{equation}
C\triangleq I_{N}-\mu\eta L.\label{eq: combination matrix C}
\end{equation}}
{Next note that}
\begin{align}
&\expec\|(I-\mu H_{k,\eta})\bz_{k,i-1}+\mu\bc_{k,i-1}\|^2\nonumber\\
&=\expec\left\|t\frac{1}{t}(I_M-\mu H_{k,\eta})\bz_{k,i-1}+\mu(1-t)\frac{1}{1-t} \bc_{k,i-1}\right\|^2\nonumber\\
&\leq \frac{1}{t}\expec\|(I_M-\mu H_{k,\eta})\bz_{k,i-1}\|^2+\mu^2\frac{1}{1-t}\expec\| \bc_{k,i-1}\|^2,
\end{align}
for any arbitrary positive number $t\in(0,1)$. We select $t=\gamma_k$ with $\gamma_k$ defined in~\eqref{eq: gamma_k}. {This gives}
\begin{equation}
\label{eq: jensen's inequality for the difference 1}
\expec\|(I-\mu H_{k,\eta})\bz_{k,i-1}+\mu\bc_{k,i-1}\|^2\leq \gamma_k\expec\|\bz_{k,i-1}\|^2+\mu^2\frac{1}{1-\gamma_k}\expec\| \bc_{k,i-1}\|^2,
\end{equation}
Let us introduce the mean-square perturbation vector at time $i$:
\begin{equation}
\text{MSP}_{z,i}\triangleq\col\left\{\expec\|\bz_{k,i}\|^2\right\}_{k=1}^N.
\end{equation}
Replacing~\eqref{eq: jensen's inequality for the difference 1} into~\eqref{eq: jensen's inequality for the difference}, and using~\eqref{eq: bound on squared ck}, we obtain:
\begin{equation}
\text{MSP}_{z,i}\preceq CG''\text{MSP}_{z,i-1}+\mu^2\frac{1}{4}(\kappa'_d)^2C(I_N-G'')^{-1}(\mathds{1}_N\otimes\expec\|\cw^o_{\eta}-\bcw_{i-1}\|^4),
\end{equation}
with $G''=\diag\{\gamma_k\}_{k=1}^N$. Iterating the above recursion starting from $i=1$, we obtain:
\begin{equation}
\text{MSP}_{z,i}\preceq (CG'')^i\text{MSP}_{z,0}+\mu^2\frac{1}{4}(\kappa'_d)^2\sum_{j=0}^{i-1}(CG'')^jC(I_N-G'')^{-1}(\mathds{1}_N\otimes\expec\|\cw^o_{\eta}-\bcw_{i-1-j}\|^4).
\end{equation}
Under Assumption~\ref{assumption: combination matrix} and condition~\eqref{eq: condition 1}, the matrix $CG''$ is guaranteed to be stable. Following similar arguments {to the ones} used to {establish eq.~(68) in Part I (Appendix E)~\cite{nassif2018diffusion}}, and from {Theorem~3 in Part I~\cite{nassif2018diffusion}}, we conclude that:
\begin{equation}
\|\limsup_{i\rightarrow\infty}\text{MSP}_{z,i}\|_{\infty}=O(\mu^2),
\end{equation}
where we used the fact that $\|(I_N-G'')^{-1}\|_{\infty}\leq O(\mu^{-1})$.  It follows that
 \begin{equation}
 \limsup_{i\rightarrow\infty}\expec\|\bcwt_i-\bcwt_i'\|^2=O(\mu^2).
 \end{equation}
 Finally note that
 \begin{align}
 \expec\|\bcwt_i'\|^2=\expec\|\bcwt'_i-\bcwt_i+\bcwt_i\|^2&\leq\expec\|\bcwt'_i-\bcwt_i\|^2+\expec\|\bcwt_i\|^2+2|\expec(\bcwt'_i-\bcwt_i)^\top\bcwt_i|\notag\\
 &\leq\expec\|\bcwt'_i-\bcwt_i\|^2+\expec\|\bcwt_i\|^2+2\sqrt{\expec\|\bcwt'_i-\bcwt_i\|^2\expec\|\bcwt_i\|^2}
 \end{align}
where we used $|\expec\bx|\leq\expec|\bx|$ from Jensen's inequality and where we applied H\"older's inequality:$$\expec|\bx^\top\by|\leq\left(\expec|\bx|^p\right)^{\frac{1}{p}}\left(\expec|\bx|^q\right)^{\frac{1}{q}}, \quad{\text{when }1/p+1/q=1}.$$ Hence we get:
 \begin{equation}
 \limsup_{i\rightarrow\infty}\left( \expec\|\bcwt_i'\|^2-\expec\|\bcwt_i\|^2\right)\leq O(\mu^2)+\sqrt{O(\mu^3)}=O(\mu^{\frac{3}{2}}),
 \end{equation}
 since $\mu^2<\mu^{\frac{3}{2}}$ for small $\mu\ll 1$.

\section{Proof of Lemma~\ref{lemma: mean-square stability long term}}
\label{app: proof of mean-square stability of the long term model}
\noindent From~\eqref{eq: recursion for w' infinity w'i}, we have:
\begin{equation}
w'_{k,\infty}-\bw'_{k,i}=\sum_{\ell=1}^N[C]_{k\ell}\bphi'_{\ell,i},
\end{equation}
where $C$ is defined in~\eqref{eq: combination matrix C} and where $\bphi'_{k,i}$ is  given by:
\begin{equation}
\bphi'_{k,i}=(I_M-\mu H_{k,\eta})(w'_{k,\infty}-\bw'_{k,i-1})-\mu\bs_{k,i}(\bw_{k,i-1}).
\end{equation}
Applying Jensen's inequality~\cite[pp.~77]{boyd2004convex} {to} the convex function $\|\cdot\|^2$, we obtain:
\begin{equation}
\label{eq: relation by Jensen's 1}
\expec\|w'_{k,\infty}-\bw'_{k,i}\|^2\leq\sum_{\ell=1}^N[C]_{k\ell}\expec\|\bphi'_{\ell,i}\|^2.
\end{equation}
Under Assumption~\ref{assumption: gradient noise}, we have:
\begin{equation}
\label{eq: relation conditioned 1}
\begin{split}
\expec[\|\bphi'_{k,i}\|^2|\bcF_{i-1}]&=\|w'_{k,\infty}-\bw'_{k,i-1}\|^2_{\Sigma_{k}}+\mu^2\expec[\|\bs_{k,i}(\bw_{k,i-1})\|^2|\bcF_{i-1}].
\end{split}
\end{equation}
where $\Sigma_{k}\triangleq(I_M-\mu H_{k,\eta})^2$, which due to Assumption~\ref{assumption: strong convexity}, can be bounded as follows:
\begin{equation}
0<\Sigma_{k}\leq\gamma_k^2I_M,
\end{equation}
where $\gamma_k$ is given by~\eqref{eq: gamma_k}. Taking expectation again in~\eqref{eq: relation conditioned 1}, and using the {bound (136) on $\expec[\|\bs_{k,i}(\bw_{k,i-1})\|^2|\bcF_{i-1}]$ from Part I~\cite{nassif2018diffusion}}, we obtain:
\begin{equation}
\label{eq: relation unconditioned 1}
\begin{split}
\expec\|\bphi'_{k,i}\|^2&=\expec\|w'_{k,\infty}-\bw'_{k,i-1}\|^2_{\Sigma_{k}}+\mu^2\expec\|\bs_{k,i}(\bw_{k,i-1})\|^2\\
&\leq\gamma_k^2\expec\|w'_{k,\infty}-\bw'_{k,i-1}\|^2+3\mu^2\beta^2_k\expec\|w_{k,\infty}-\bw_{k,i-1}\|^2+\mu^2\left(3\beta_k^2\|w^o_{k,\eta}-w_{k,\infty}\|^2+3\beta_k^2\|w^o_{k,\eta}\|^2+\sigma^2_{s,k}\right).
\end{split}
\end{equation}
Now, combining~\eqref{eq: relation unconditioned 1} and~\eqref{eq: relation by Jensen's 1}, we obtain~\eqref{eq: evolution of the MSP' i}.

Iterating~\eqref{eq: evolution of the MSP' i} starting from $i=1$, we get:
\begin{equation}
\label{eq: recursion MSP 1}
\text{MSP}'_i \preceq (C(G'')^2)^i\text{MSP}'_{0}+\mu^2\sum_{j=0}^{i-1}(C(G'')^2)^jC(3\diag\{\beta^2_k\}\text{MSP}_{i-1-j}+b).
\end{equation}
Under Assumption~\ref{assumption: combination matrix} and condition~\eqref{eq: condition 1}, the matrix $C(G'')^2$ is guaranteed to be stable. Using the fact that $b=O(1)$ from {Theorem~1 in Part I~\cite{nassif2018diffusion}}, $1-\|(G'')^2\|_{\infty}=O(\mu)$, and that $\|\lim_{i\rightarrow\infty}\text{MSP}_i\|_{\infty}=O(\mu)$ from {Theorem~2 in Part I~\cite{nassif2018diffusion}}, and following similar arguments as the one used to {establish eq.~(68) in Part I (Appendix E)~\cite{nassif2018diffusion}}, we conclude~\eqref{eq: steady-state of the MSP'}.

From~\eqref{eq: bound on mean-square expectation long term}, Lemma~\ref{lemma: mean stability long term}, and~\eqref{eq: steady-state of the MSP'}, we conclude~\eqref{eq: steady-state of the second order long term}.

\section{Limiting second-order moment of gradient noise}
\label{app: proof of limiting second order of gradient noise}
\noindent It is shown in~\cite[Sec.~4.1]{sayed2014adaptation} that the conditional noise covariance matrix satisfying Assumptions~\ref{assumption: gradient noise} and~\ref{assumption: Smoothness condition on noise covariance} satisfies more globally a condition of the following form for all $\bw_k\in\bcF_{i-1}$:
\begin{equation}
\label{eq: global smoothness condition on noise covariance}
\|R_{s,k,i}(\bw_k)-R_{s,k,i}(w^o_{k,\eta})\|\leq\kappa_d\|w^o_{k,\eta}-\bw_k\|^{{\theta}}+\kappa'_d\|w^o_{k,\eta}-\bw_k\|^2,
\end{equation}
for some nonnegative constant $\kappa'_d$. By adding and subtracting the same term we have:
\begin{equation}
\begin{split}
\expec[\bs_{k,i}(\bw_{k,i-1})\bs_{k,i}^\top(\bw_{k,i-1})|\bcF_{i-1}]=&\expec[\bs_{k,i}(w_{k,\eta}^o)\bs_{k,i}^\top(w_{k,\eta}^o)|\bcF_{i-1}]+\expec[\bs_{k,i}(\bw_{k,i-1})\bs_{k,i}^\top(\bw_{k,i-1})|\bcF_{i-1}]-\\
&\qquad\expec[\bs_{k,i}(w_{k,\eta}^o)\bs_{k,i}^\top(w_{k,\eta}^o)|\bcF_{i-1}]
\end{split}
\end{equation}
which using definition~\eqref{eq: R ski definition} can be rewritten as:
\begin{equation}
\expec[\bs_{k,i}(\bw_{k,i-1})\bs_{k,i}^\top(\bw_{k,i-1})|\bcF_{i-1}]=\expec[\bs_{k,i}(w_{k,\eta}^o)\bs_{k,i}^\top(w_{k,\eta}^o)|\bcF_{i-1}]+R_{s,k,i}(\bw_{k,i-1})-R_{s,k,i}(w_{k,\eta}^o).
\end{equation}
Subtracting $R_{s,k,\eta}$ defined by~\eqref{eq: R sk definition} from both sides and computing expectations, we get:
\begin{equation}
\expec\bs_{k,i}(\bw_{k,i-1})\bs_{k,i}^\top(\bw_{k,i-1})-R_{s,k,\eta}=\expec\left(\expec[\bs_{k,i}(w_{k,\eta}^o)\bs_{k,i}^\top(w_{k,\eta}^o)|\bcF_{i-1}]-R_{s,k,\eta}\right)+\expec\left(R_{s,k,i}(\bw_{k,i-1})-R_{s,k,i}(w_{k,\eta}^o)\right).
\end{equation}
It then follows from the triangle inequality of norms, and from Jensen's inequality that:
\begin{align}
&\|\expec\bs_{k,i}(\bw_{k,i-1})\bs_{k,i}^\top(\bw_{k,i-1})-R_{s,k,\eta}\|\nonumber\\
&\leq\left\|\expec\left(\expec[\bs_{k,i}(w_{k,\eta}^o)\bs_{k,i}^\top(w_{k,\eta}^o)|\bcF_{i-1}]-R_{s,k,\eta}\right)\right\|+\|\expec\left(R_{s,k,i}(\bw_{k,i-1})-R_{s,k,i}(w_{k,\eta}^o)\right)\|\nonumber\\
&\leq\expec\|\expec[\bs_{k,i}(w_{k,\eta}^o)\bs_{k,i}^\top(w_{k,\eta}^o)|\bcF_{i-1}]-R_{s,k,\eta}\|+\expec\|R_{s,k,i}(\bw_{k,i-1})-R_{s,k,i}(w_{k,\eta}^o)\|.
\end{align}
Computing the limit superior of both sides and using~\eqref{eq: R sk definition} to annihilate the limit of the first term on the {RHS}, we conclude that:
\begin{equation}
\label{eq: relation 1 on the covariance}
\limsup_{i\rightarrow\infty}\|\expec\bs_{k,i}(\bw_{k,i-1})\bs_{k,i}^\top(\bw_{k,i-1})-R_{s,k,\eta}\|\leq\limsup_{i\rightarrow\infty}\expec\|R_{s,k,i}(\bw_{k,i-1})-R_{s,k,i}(w_{k,\eta}^o)\|.
\end{equation}
We next use the smoothness condition~\eqref{eq: global smoothness condition on noise covariance} to bound the right-most term as follows:
\begin{equation}
\|R_{s,k,i}(\bw_{k,i-1})-R_{s,k,i}(w_{k,\eta}^o)\|\leq\kappa_d\left(\|\bwt_{k,i-1}\|^4\right)^{\frac{{{\theta}}}{4}}+\kappa'_d\|\bwt_{k,i-1}\|^2.
\end{equation}
Under expectation and in the limit, we have:
\begin{align}
\limsup_{i\rightarrow\infty}\expec\|R_{s,k,i}(\bw_{k,i-1})-R_{s,k,i}(\bw_{k,\eta}^o)\|&\leq\limsup_{i\rightarrow\infty}\left\{\kappa_d\expec\left(\|\bwt_{k,i-1}\|^4\right)^{\frac{{{\theta}}}{4}}+\kappa'_d\expec\|\bwt_{k,i-1}\|^2\right\}\nonumber\\
&\leq\limsup_{i\rightarrow\infty}\left\{\kappa_d\left(\expec\|\bwt_{k,i-1}\|^4\right)^{\frac{{{\theta}}}{4}}+\kappa'_d\expec\|\bwt_{k,i-1}\|^2\right\}=O(\mu^{\frac{{{\theta}}'}{2}}),
\end{align}
where we applied Jensen's inequality to the function $f(x)=x^{\frac{{{\theta}}}{4}}$; this function is concave over $x\geq 0$ for ${{\theta}}\in(0,4]$. Moreover, in the last step we called upon the results in {Theorems~2 and 3 in Part I~\cite{nassif2018diffusion}} where it is shown that the second and fourth order moments of $\bwt_{k,i-1}$ are asymptotically bounded by $O(\mu)$ and $O(\mu^2)$, respectively. Accordingly, the exponent ${{\theta}}'\triangleq\min\{{{\theta}},2\}$ since $O(\mu^{{{\theta}}/2})$ dominates $O(\mu)$ for values of ${{\theta}}\in(0,2]$ and $O(\mu)$ dominates $O(\mu^{{{\theta}}/2})$ for values of ${{\theta}}\in[2,4]$. Substituting into~\eqref{eq: relation 1 on the covariance} gives:
\begin{equation}
\limsup_{i\rightarrow\infty}\|\expec\bs_{k,i}(\bw_{k,i-1})\bs_{k,i}^\top(\bw_{k,i-1})-R_{s,k,\eta}\|=O(\mu^{\frac{{{\theta}}'}{2}}),
\end{equation}
which leads to~\eqref{eq: relation on the covariance}.
\section{Proof of Lemma~\ref{lemm: coefficient matrix}}
\label{app: coefficient matrix F eta}
\noindent Consider the matrix $\overline{\cF}_{\eta}$ in~\eqref{eq: transformed matrix F_eta}. Using the block Kronecker product property:
\begin{equation}
\label{eq: block property}
(\cA\otimes_b\cB)(\cC\otimes_b\cD)=(\cA\cC\otimes_b\cB\cD),
\end{equation}
it can be verified that:
\begin{equation}
\label{eq: transformed F_eta}
\overline{\cF}_{\eta}=\overline{\cB}_{\eta}^{\top}\otimes_b\overline{\cB}_{\eta}^{\top},
\end{equation}
where
\begin{equation}
\label{eq: transformed B_eta}
\overline{\cB}_{\eta}=\cV^{\top}\cB_{\eta}\cV=(I_{MN}-\mu\eta\cJ)(I_{MN}-\mu\cV^{\top}\cH_{\eta}\cV)
\end{equation}
so that:
\begin{equation}
\overline{\cB}_{\eta}^{\top}=\left[
\begin{array}{cccc}
I-\mu H_{11}&-\mu(1-\mu\eta\lambda_2)H_{21}&\ldots&-\mu(1-\mu\eta\lambda_N)H_{N1}\\
-\mu H_{12}&(1-\mu\eta\lambda_2)(I-\mu H_{22})&\ldots&-\mu(1-\mu\eta\lambda_N)H_{N2}\\
\vdots&\vdots&&\vdots\\
-\mu H_{1N}&-\mu(1-\mu\eta\lambda_2)H_{2N}&\ldots&(1-\mu\eta\lambda_N)(I-\mu H_{NN})
\end{array}
\right],
\end{equation}
where $H_{mn}$ is defined in~\eqref{eq: H mn}. It can be verified that the matrix $Z=I-\overline{\cF}_{\eta}$ is $N\times N$ blocks $Z_{mn}$ with each block of size $M^2N\times M^2N$:
\begin{equation}
\label{eq: Z}
Z=I-\overline{\cF}_{\eta}=\left[
\begin{array}{cccc}
Z_{11}&Z_{12}&\ldots&Z_{1N}\\
Z_{21}&Z_{22}&\ldots&Z_{2N}\\
\vdots&\vdots&&\vdots\\
Z_{N1}&Z_{N2}&\ldots&Z_{NN}
\end{array}
\right],
\end{equation}
We denote by $[Z_{mn}]_{pq}$ the $M^2\times M^2$ $(p,q)$-th block of  $Z_{mn}$. We have:
\begin{equation}
\label{eq: Z blocks}
[Z_{mn}]_{pq}=
\left\lbrace
\begin{array}{lr}
I-(1-\mu\eta\lambda_m)(1-\mu\eta\lambda_p)[(I-\mu H_{mm})\otimes(I-\mu H_{pp})],&\text{if }m=n,~p=q\\
\mu(1-\mu\eta\lambda_m)(1-\mu\eta\lambda_q)[(I-\mu H_{mm})\otimes H_{qp}],&\text{if }m=n,~p\neq q\\
\mu(1-\mu\eta\lambda_n)(1-\mu\eta\lambda_p)[H_{nm}\otimes(I-\mu H_{pp})],&\text{if }m\neq n,~p= q\\
-\mu^2(1-\mu\eta\lambda_n)(1-\mu\eta\lambda_q)[H_{nm}\otimes H_{qp}],&\text{if }m\neq n,~p\neq q
\end{array}
\right.
\end{equation}
We have:
\begin{equation}
(I-\mu H_{mm})\otimes(I-\mu H_{pp})=I-\mu H_{mm}\oplus H_{pp}+\mu^2H_{mm}\otimes H_{pp},
\end{equation}
where $H_{mm}\oplus H_{pp}$ is given by~\eqref{eq: H m p plus} and
\begin{equation}
\label{eq: relation on 1 -product}
1-(1-\mu\eta\lambda_m)(1-\mu\eta\lambda_p)=\mu\eta(\lambda_m+\lambda_p-\mu\eta\lambda_m\lambda_p).
\end{equation}
Thus,
\begin{equation}
\label{eq: Z blocks 1}
[Z_{mn}]_{pq}=\mu\cdot
\left\lbrace
\begin{array}{lr}
(1-\mu\eta\lambda_m)(1-\mu\eta\lambda_p)(H_{mm}\oplus H_{pp})+\eta(\lambda_m+\lambda_p-\mu\eta\lambda_m\lambda_p)I+O(\mu),&\text{if }m=n,~p=q\\
(1-\mu\eta\lambda_m)(1-\mu\eta\lambda_q)(I\otimes H_{qp})+O(\mu),&\text{if }m=n,~p\neq q\\
(1-\mu\eta\lambda_n)(1-\mu\eta\lambda_p)(H_{nm}\otimes I)+O(\mu),&\text{if }m\neq n,~p= q\\
-\mu(1-\mu\eta\lambda_n)(1-\mu\eta\lambda_q)[H_{nm}\otimes H_{qp}]=O(\mu),&\text{if }m\neq n,~p\neq q
\end{array}
\right.
\end{equation}
Before proceeding, we recall the following useful properties of the Kronecker and Kronecker sum products~\cite{bernstein2005matrix}. Let $\{\lambda_i(A),i=1,\ldots,M\}$ and $\{\lambda_j(B),j=1,\ldots,M\}$ denote the eigenvalues of any two $M\times M$ matrices $A$ and $B$, respectively. Then,
\begin{align}
\{\lambda(A\otimes B)\}&=\{\lambda_i(A)\lambda_j(B)\}_{i=1,j=1}^{M,M},\\
\{\lambda(A\oplus B)\}&=\{\lambda_i(A)+\lambda_j(B)\}_ {i=1,j=1}^{M,M}.\label{eq: property of the kronecker sum}
\end{align}
From~\eqref{eq: Z},~\eqref{eq: Z blocks}, and~\eqref{eq: Z blocks 1}, it can be verified that the matrix $Z$ can be written as:
\begin{equation}
\label{eq: Z as the sum of two matrices}
Z=X+Y.
\end{equation}
The matrix $X$ is $N^2\times N^2$ block diagonal defined as:
\begin{equation}
\label{eq: X}
X\triangleq\mu\cdot\diag\left\{\diag\left\{ [Z_{mm}]_{pp}\right\}_{p=1}^N\right\}_{m=1}^N=\mu\cdot\left[
\begin{array}{cc}
O(1)&0\\
0&O(1)+O(\eta)
\end{array}
\right],
\end{equation}
where we used the fact that for $m=1$ and $p=1$, we have $ [Z_{11}]_{11}=\mu\cdot H_{11}\oplus H_{11}+O(\mu^2)$ which is $\mu\cdot O(1)$. This is due to property~\eqref{eq: property of the kronecker sum} and the fact that $H_{11}=\frac{1}{N}\sum_{k=1}^NH_{k,\eta}>0$. For the remaining blocks of $X$, from~\eqref{eq: relation on 1 -product}, property~\eqref{eq: property of the kronecker sum}, and the fact that $H_{mm}=\sum_{k=1}^N[v_m]_k^2H_{k,\eta}>0$, it can be verified that the matrix:
\begin{equation}
\label{eq: Z_mm pp}
[Z_{mm}]_{pp}=\mu(1-\mu\eta\lambda_m)(1-\mu\eta\lambda_p)(H_{mm}\oplus H_{pp})+\mu\eta(\lambda_m+\lambda_p-\mu\eta\lambda_m\lambda_p)I
\end{equation}
is also positive definite when:
\begin{equation}
0<(1-\mu\eta\lambda_m)(1-\mu\eta\lambda_p)\leq1.
\end{equation}
Furthermore, in this case, we have $[Z_{mm}]_{pp}=\mu\cdot (O(1)+O(\eta))$. The matrix $Y=Z-X$ in~\eqref{eq: Z as the sum of two matrices} is an $N^2\times N^2$ block matrix where each block is $M^2\times M^2$ given by:
\begin{equation}
\label{eq: Y blocks}
[Y_{mn}]_{pq}=\mu\cdot
\left\lbrace
\begin{array}{lr}
0,&\text{if }m=n,~p=q\\
(1-\mu\eta\lambda_m)(1-\mu\eta\lambda_q)(I\otimes H_{qp})+O(\mu)\leq O(1),&\text{if }m=n,~p\neq q\\
(1-\mu\eta\lambda_n)(1-\mu\eta\lambda_p)(H_{nm}\otimes I)+O(\mu)\leq O(1),&\text{if }m\neq n,~p= q\\
-\mu(1-\mu\eta\lambda_n)(1-\mu\eta\lambda_q)[H_{nm}\otimes H_{qp}]\leq O(\mu),&\text{if }m\neq n,~p\neq q
\end{array}
\right.
\end{equation}
Applying the matrix inversion {identity~\cite{kailath1980linear}}, we obtain:
\begin{equation}
(X+Y)^{-1}=X^{-1}-X^{-1}Y(I+X^{-1}Y)^{-1}X^{-1},
\end{equation}
From~\eqref{eq: X}, we have:
\begin{equation}
\label{eq: inverse X}
X^{-1}=\mu^{-1}\cdot\diag\left\{\diag\left\{ ([Z_{mm}]_{pp})^{-1}\right\}_{p=1}^N\right\}_{m=1}^N=\mu^{-1}\cdot\left[
\begin{array}{cc}
O(1)&0\\
0&(O(1)+O(\eta))^{-1}
\end{array}
\right].
\end{equation}
From~\eqref{eq: inverse X} and~\eqref{eq: Y blocks}, we have:
\begin{equation}
\label{eq: XY}
X^{-1}Y=\left[
\begin{array}{cc}
0&O(1)\\
(O(1)+O(\eta))^{-1}&(O(1)+O(\eta))^{-1}
\end{array}
\right],
\end{equation}
and
\begin{equation}
I+X^{-1}Y=\left[
\begin{array}{cc}
O(1)&O(1)\\
(O(1)+O(\eta))^{-1}&O(1)
\end{array}
\right].
\end{equation}
Applying the block inversion formula to $I+X^{-1}Y$, we obtain:
\begin{equation}
\label{eq: inverse I+XY}
(I+X^{-1}Y)^{-1}=\left[
\begin{array}{cc}
O(1)&O(1)\\
(O(1)+O(\eta))^{-1}&O(1)
\end{array}
\right].
\end{equation}
Finally, from~\eqref{eq: inverse X},~\eqref{eq: XY}, and~\eqref{eq: inverse I+XY}, we conclude that:
\begin{equation}
X^{-1}Y(I+X^{-1}Y)^{-1}X^{-1}=\mu^{-1}\cdot\left[
\begin{array}{cc}
(O(1)+O(\eta))^{-1}&(O(1)+O(\eta))^{-1}\\
(O(1)+O(\eta))^{-1}&(O(1)+O(\eta))^{-2}
\end{array}
\right].
\end{equation}
Consider now the matrix $(I-\cF_{\eta})^{-1}=(\cV\otimes_b\cV)(I-\overline{\cF}_{\eta})^{-1}(\cV\otimes_b\cV)^{\top}$. It can be verified that:
\begin{align}
(I-\cF_{\eta})^{-1}&=(\cV\otimes_b\cV)X^{-1}(\cV\otimes_b\cV)^{\top}+(\cV\otimes_b\cV)X^{-1}Y(I+X^{-1}Y)^{-1}X^{-1}(\cV\otimes_b\cV)^{\top}\nonumber\\
&=\sum_{m=1}^N\sum_{p=1}^N[(v_m\otimes I_M)\otimes_b(v_p\otimes I_M)]([Z_{mm}]_{pp})^{-1}[(v^{\top}_m\otimes I_M)\otimes_b(v_p^{\top}\otimes I_M)]+\mu^{-1}(O(1)+O(\eta))^{-1}.\label{eq: inverse of I- F eta}
\end{align}

\section{Proof of Lemma~\ref{lemma: steady-state performance}}
\label{app: steady-state performance}
\noindent Recursion~\eqref{eq: long term error model 1} for the long term model includes a constant driving term on the {RHS} represented by $\mu^2\eta^2\cL^2\cw^o_{\eta}$. To facilitate the variance analysis, we introduce the centered variable:
\begin{equation}
\bz_i\triangleq\bcwt'_{i}-\expec\,\bcwt'_{i}.
\end{equation}
Subtracting~\eqref{eq: mean recursion of the long term model} from~\eqref{eq: long term error model 1}, we obtain:
\begin{equation}
\label{eq: recursion centered variable}
\bz_i=\cB_{\eta}\bz_{i-1}-\mu(I_{MN}-\mu\eta\cL)\bs_{i}(\bcw_{i-1}),
\end{equation}
where the deterministic driving terms are removed. Although we are interested in evaluating the asymptotic size of $\expec\|\bcwt'_{i}\|^2$, we can rely on the centered variable $\bz_i$ for this purpose, since, from Lemma~\ref{lemma: mean stability long term}, it holds for $i\gg 1$:
\begin{equation}
\expec\|\bz_{i}\|^2=\expec\|\bcwt'_{i}\|^2-\|\expec\bcwt'_{i}\|^2=\expec\|\bcwt'_{i}\|^2+O(\mu^2),
\end{equation}
for fixed $\eta$. Furthermore, we established in Lemma~\ref{lemma: dimension of the approximation error}, that the error variances $\expec\|\bcwt'_{i}\|^2$ and $\expec\|\bcwt_{i}\|^2$ are within $O(\mu^{\frac{3}{2}})$ from each other. Therefore, we may evaluate the mean-square error in terms of the mean-square value of the variable $\bz_i$ by employing the correction:
\begin{equation}
\label{eq: mse of the centered variable}
\limsup_{i\rightarrow \infty}\frac{1}{N}\expec\|\bcwt_{i}\|^2=\limsup_{i\rightarrow \infty}\frac{1}{N}\expec\|\bz_{i}\|^2+O(\mu^{\frac{3}{2}}).
\end{equation}
We therefore continue with recursion~\eqref{eq: recursion centered variable} and proceed to examine how the mean-square value of $\bz_i$ evolves over time by relying on the energy conservation arguments.

Let $\Sigma$ denote an arbitrary symmetric positive semi-definite matrix that we are free to choose. Equating the squared weighted values of both sides of~\eqref{eq: recursion centered variable} and taking expectations conditioned on the past history gives:
\begin{equation}
\expec\left[\|\bz_i\|_{\Sigma}^2|\bcF_{i-1}\right]=\|\bz_{i-1}\|^2_{\cB_{\eta}^\top\Sigma\cB_{\eta}}+\mu^2\expec\left[\|\bs_{i}\|^2_{(I_{MN}-\mu\eta\cL)\Sigma(I_{MN}-\mu\eta\cL)}|\bcF_{i-1}\right]
\end{equation}
Taking expectation again removes the conditioning and we get:
\begin{equation}
\label{eq: mean-square error z_i}
\expec\|\bz_i\|_{\Sigma}^2=\expec\left(\|\bz_{i-1}\|^2_{\cB_{\eta}^\top\Sigma\cB_{\eta}}\right)+\mu^2\expec\left(\|\bs_{i}\|^2_{(I_{MN}-\mu\eta\cL)\Sigma(I_{MN}-\mu\eta\cL)}\right).
\end{equation}
Consider the right-most term. We have:
\begin{equation}
\label{eq: gradient noise covariance weighted}
\mu^2\expec\left(\|\bs_{i}\|^2_{(I_{MN}-\mu\eta\cL)\Sigma(I_{MN}-\mu\eta\cL)}\right)=\mu^2\tr\left[(I_{MN}-\mu\eta\cL)\Sigma(I_{MN}-\mu\eta\cL)\expec\left(\bs_i(\bcw_{i-1})\bs_i^{\top}(\bcw_{i-1})\right)\right].
\end{equation}
Using~\eqref{eq: relation on the covariance} and the fact that the gradient noises across the agents are uncorrelated {under condition~\eqref{eq: uncorrelated gradient noises}}, we obtain:
\begin{equation}
\limsup_{i\rightarrow\infty}\|\expec\left(\bs_i(\bcw_{i-1})\bs_i^{\top}(\bcw_{i-1})\right)-\cS_\eta\|=O(\mu^{\min\left\{1,\frac{{{\theta}}}{2}\right\}})
\end{equation}
Using the sub-multiplicative property of the $2-$induced norm, we conclude that:
\begin{equation}
\limsup_{i\rightarrow\infty}\mu^2\|(I_{MN}-\mu\eta\cL)\Sigma(I_{MN}-\mu\eta\cL)\left(\expec(\bs_i(\bcw_{i-1})\bs_i^{\top}(\bcw_{i-1})-\cS_\eta\right)\|=\tr(\Sigma)\cdot O(\mu^{2+\min\left\{1,\frac{{{\theta}}}{2}\right\}}),
\end{equation}
where we used the fact that $\|\Sigma\|\leq\tr(\Sigma)$ for any positive semi-definite $\Sigma$. Using the fact that $|\tr(X)|\leq c\|X\|$ for any square matrix $X$, we obtain:
\begin{equation}
\limsup_{i\rightarrow\infty}|\mu^2\expec\|\bs_{i}\|^2_{(I_{MN}-\mu\eta\cL)\Sigma(I_{MN}-\mu\eta\cL)}-\tr(\Sigma\cY)|=\tr(\Sigma)\cdot O(\mu^{2+\min\left\{1,\frac{{{\theta}}}{2}\right\}})=b_1,
\end{equation}
$b_1=\tr(\Sigma)\cdot O(\mu^{2+\min\left\{1,\frac{{{\theta}}}{2}\right\}})\geq 0$. The above relation then implies that, given $\epsilon >0$, there exists an $i_o$ large enough such that for all $i>i_o$ it holds that
\begin{equation}
|\mu^2\expec\|\bs_{i}\|^2_{(I_{MN}-\mu\eta\cL)\Sigma(I_{MN}-\mu\eta\cL)}-\tr(\Sigma\cY)|\leq b_1+\epsilon.
\end{equation}
If we select $\epsilon=\tr(\Sigma)\cdot O(\mu^{2+\min\left\{1,\frac{{{\theta}}}{2}\right\}})$ and introduce the sum $b_o=b_1+\epsilon$, then we arrive at:
\begin{equation}
\label{eq: bounds on gradient noise covariance weighted}
\tr(\Sigma\cY)-b_o\leq\mu^2\expec\|\bs_{i}\|^2_{(I_{MN}-\mu\eta\cL)\Sigma(I_{MN}-\mu\eta\cL)}\leq\tr(\Sigma\cY)+b_o,
\end{equation}
for some non-negative constant $b_o=\tr(\Sigma)\cdot O(\mu^{2+\min\left\{1,\frac{{{\theta}}}{2}\right\}})$. Substituting~\eqref{eq: bounds on gradient noise covariance weighted} into~\eqref{eq: mean-square error z_i} we obtain for $i\gg 1$:
\begin{equation}
\expec\|\bz_i\|_{\Sigma}^2\leq\expec\|\bz_{i-1}\|^2_{\cB_{\eta}^\top\Sigma\cB_{\eta}}+\tr(\Sigma\cY)+b_o.
\end{equation}
Using the sub-additivity  property of the limit superior, we obtain:
\begin{equation}
\limsup_{i\rightarrow\infty}\expec\|\bz_i\|_{\Sigma}^2\leq\limsup_{i\rightarrow\infty}\expec\|\bz_{i-1}\|^2_{\cB_{\eta}^\top\Sigma\cB_{\eta}}+\tr(\Sigma\cY)+b_o.
\end{equation}
Grouping terms we get:
\begin{equation}
\limsup_{i\rightarrow\infty}\expec\|\bz_i\|_{\Sigma-{\cB_{\eta}^\top\Sigma\cB_{\eta}}}^2\leq\tr(\Sigma\cY)+b_o,
\end{equation}
We conclude that the limit superior of the error variance satisfies:
\begin{equation}
\label{eq: limit superior of the error variance}
\limsup_{i\rightarrow\infty}\expec\|{\bz}_i\|_{\Sigma-{\cB_{\eta}^\top\Sigma\cB_{\eta}}}^2=\tr(\Sigma\cY)+\tr(\Sigma)O(\mu^{2+\min\left\{1,\frac{{{\theta}}}{2}\right\}}),
\end{equation}

In order to obtain identity as a weighting matrix on the mean-square value of ${\bz}_i$ in~\eqref{eq: limit superior of the error variance}, we select $\Sigma$ as the solution to the following discrete time Lyapunov equation:
\begin{equation}
\label{eq: choice of Sigma}
\Sigma-{\cB_{\eta}^\top\Sigma\cB_{\eta}}=I_{MN}.
\end{equation}
We know that $\cB_{\eta}$ is stable under conditions~\eqref{eq: condition for stability}  and~\eqref{eq: condition 1}. Accordingly, we are guaranteed that the above Lyapunov equation has a unique solution $\Sigma$, and moreover, this solution is symmetric and non-negative definite as desired. We can then focus on evaluating the {RHS} of~\eqref{eq: limit superior of the error variance}.

For this purpose, we start by applying the block vectorization operation to both sides of~\eqref{eq: choice of Sigma} to find that:
\begin{equation}
\bvc(\Sigma)-(\cB_{\eta}^\top\otimes_b\cB_{\eta}^\top)\bvc(\Sigma)=\bvc(I_{MN}),
\end{equation}
so that in terms of the matrix $\cF_{\eta}$ defined in~\eqref{eq: matrix F_eta}, we can write:
\begin{equation}
\bvc(\Sigma)=(I-\cF_{\eta})^{-1}\bvc(I_{MN}).
\end{equation}
Now, substituting this $\Sigma$ into~\eqref{eq: limit superior of the error variance}, we obtain $\expec\|\bz_i\|^2$ on the left-hand side while the term $\tr(\Sigma\cY)$ on the {RHS} becomes:
\begin{equation}
\label{eq: trace sigma Y}
\tr(\Sigma\cY)=[\bvc(\cY^{\top})]^{\top}(I-\cF_{\eta})^{-1}\bvc(I_{MN}).
\end{equation}
Likewise the second term on the {RHS} of~\eqref{eq: limit superior of the error variance} becomes:
\begin{equation}
\label{eq: higher order term}
O(\mu^{2+\min\left\{1,\frac{{{\theta}}}{2}\right\}})\cdot\tr(\Sigma)=O(\mu^{2+\min\left\{1,\frac{{{\theta}}}{2}\right\}})\cdot[\bvc(I_{MN})]^{\top}(I-\cF_{\eta})^{-1}\bvc(I_{MN}).
\end{equation}
We now verify that $|[\bvc(I_{MN})]^{\top}(I-\cF_{\eta})^{-1}\bvc(I_{MN})|=O(1/\mu)$. This result will permit us to assess the size of the second term on the  {RHS} of~\eqref{eq: limit superior of the error variance}. We have:
\begin{equation}
\label{eq: size of higher order term}
|[\bvc(I)]^{\top}(I-\cF_{\eta})^{-1}\bvc(I)|\leq\|\bvc(I)\|^2\cdot\|(I-\cF_{\eta})^{-1}\|\leq r\cdot\|(I-\cF_{\eta})^{-1}\|_1\|\bvc(I)\|^2=O(\mu^{-1}),
\end{equation}
where we used a positive constant $r$ to account for the fact that matrix norms are equivalent.

Returning to~\eqref{eq: limit superior of the error variance}, and using~\eqref{eq: choice of Sigma},~\eqref{eq: trace sigma Y},~\eqref{eq: higher order term}, and~\eqref{eq: size of higher order term}, we conclude that:
\begin{equation}
\limsup_{i\rightarrow\infty}\expec\|{\bz}_i\|^2=[\bvc(\cY^{\top})]^{\top}(I-\cF_{\eta})^{-1}\bvc(I_{MN})+O(\mu^{1+\min\left\{1,\frac{{{\theta}}}{2}\right\}}),
\end{equation}
with  ${{\theta}}\in(0,4]$. But since $\cF_{\eta}$ is a stable matrix, we can employ the expansion:
\begin{equation}
(I-\cF_{\eta})^{-1}=\sum_{n=0}^{\infty}\cF_{\eta}^n=\sum_{n=0}^{\infty}(\cB_{\eta}^\top)^n\otimes_b(\cB_{\eta}^\top)^n,
\end{equation}
and write:
\begin{equation}
[\bvc(\cY^{\top})]^{\top}(I-\cF_{\eta})^{-1}\bvc(I_{MN})=\sum_{n=0}^{\infty}\tr(\cB_{\eta}^n\cY(\cB_{\eta}^\top)^n),
\end{equation}
This series converges to the trace value of the unique solution of the following Lyapunov equation:
\begin{equation}
\cX-\cB_{\eta}\cX\cB_{\eta}^\top=\cY,
\end{equation}
where
\begin{equation}
\label{eq: X and X'}
\cX=\sum_{n=0}^{\infty}\cB_{\eta}^n\cY(\cB_{\eta}^\top)^n.
\end{equation}
Consequently, using~\eqref{eq: mse of the centered variable}, we obtain:
\begin{equation}
\label{eq: global mse ss}
\begin{split}
\limsup_{i\rightarrow \infty}\frac{1}{N}\expec\|\bcwt_{i}\|^2&=\frac{1}{N}\tr(\cX)+O(\mu^{1+{{\theta}}_m})\\
&=\frac{1}{N}\sum_{n=0}^{\infty}\tr(\cB_{\eta}^n\cY(\cB_{\eta}^{\top})^n)+O(\mu^{1+{{\theta}}_m}),
\end{split}
\end{equation}
where ${{\theta}}_m=\frac{1}{2}\min\{1,{{\theta}}\}$ and  where $\tr(\cX)=[\bvc(\cY^{\top})]^{\top}(I-\cF_{\eta})^{-1}\bvc(I_{MN})$ which is $O(\mu)$ since $\|\cY\|=O(\mu^2)$ and $\|(I-\cF_{\eta})^{-1}\|=O(\mu^{-1})$. Therefore the value of $\tr(\cX)$ is $O(\mu)$, which dominates the factor $O(\mu^{1+{{\theta}}_m})$.

\section{Proof of Theorem~\ref{lemma: steady-state MSD performance}}
\label{app: steady-state MSD performance}
\noindent From~\eqref{eq: steady-state network performance} and~\eqref{eq: inverse of I- F eta}, we have:
\begin{equation}
\begin{split}
&(\bvc(\cY^{\top}))^{\top}(I-\cF_{\eta})^{-1}\bvc(I_{MN})=O(\mu)(O(1)+O(\eta))^{-1}+\\
&\qquad\sum_{m=1}^N\sum_{p=1}^N(\bvc(\cY^{\top}))^{\top}[(v_m\otimes I_M)\otimes_b(v_p\otimes I_M)]([Z_{mm}]_{pp})^{-1}[(v^{\top}_m\otimes I_M)\otimes_b(v_p^{\top}\otimes I_M)]\bvc(I_{MN})
\end{split}
\end{equation}
where the $\bvc$ operation is relative to blocks of size $M\times M$. Using the property $\bvc(\cA\cC\cB)=(\cB^{\top}\otimes_b\cA)\bvc(\cC)$, we obtain:
\begin{equation}
[(v^{\top}_m\otimes I_M)\otimes_b(v_p^{\top}\otimes I_M)]\bvc(I_{MN})=\left\lbrace
\begin{array}{lr}
\bvc(I_M)=\vc(I_M),& \text{if } m=p\\
0, &\text{if  } m\neq p
\end{array}
\right.
\end{equation}
and we conclude that:
\begin{equation}
\label{eq: relation expression}
\begin{split}
&(\bvc(\cY^{\top}))^{\top}(I-\cF_{\eta})^{-1}\bvc(I_{MN})=O(\mu)(O(1)+O(\eta))^{-1}+\sum_{m=1}^N(\bvc(\cY^{\top}))^{\top}[(u_m\otimes I_M)\otimes_b(u_m\otimes I_M)]x_m
\end{split}
\end{equation}
where
\begin{equation}
x_m\triangleq([Z_{mm}]_{mm})^{-1}\vc(I_M).
\end{equation}
This vector is the unique solution to the linear system of equations:
\begin{equation}
[Z_{mm}]_{mm}x_m=\vc(I_M),
\end{equation}
or, equivalently, by using~\eqref{eq: Z_mm pp}:
\begin{equation}
\begin{split}
&\mu\left[(1-\mu\eta\lambda_m)^2(H_{mm}\otimes I)+\frac{\eta}{2}\lambda_m(2-\mu\eta\lambda_m)\right]x_m\\
&\qquad+\mu\left[(1-\mu\eta\lambda_m)^2(I\otimes H_{mm})+\frac{\eta}{2}\lambda_m(2-\mu\eta\lambda_m)I\right]x_m=\vc(I_M),
\end{split}
\end{equation}
Let $X_m=\text{unvec}(x_m)$. Applying the property $\vc(ACB)=(B^{\top}\otimes A)\vc(C)$, we obtain:
\begin{equation}
\vc(X_mT_m)+\vc(T_mX_m)=\vc(I_M)
\end{equation}
where
\begin{equation}
T_m\triangleq\mu(1-\mu\eta\lambda_m)^2H_{mm}+\frac{\mu\eta}{2}\lambda_m(2-\mu\eta\lambda_m)I.
\end{equation}
We conclude from the above equation that $X_m$ is the unique solution to the continuous time Lyapunov equation:
\begin{equation}
X_mT_m+T_mX_m=I_M,
\end{equation}
whose solution is given by:
\begin{equation}
X_m=\frac{1}{2}T_m^{-1}=\frac{1}{2\mu}\left((1-\mu\eta\lambda_m)^2H_{mm}+\frac{\eta}{2}\lambda_m(2-\mu\eta\lambda_m)I\right)^{-1}.
\end{equation}
Using the definitions~\eqref{eq: definition cY},~\eqref{eq: definition cS}, and applying properties
\begin{equation}
\bvc(\cA\cC\cB)=(\cB^{\top}\otimes_b\cA)\bvc(\cC),\quad\text{and}\quad\tr(\cA\cB)=(\bvc(\cB^{\top}))^\top\bvc(\cA),
\end{equation} we get:
\begin{align}
&\sum_{m=1}^N[\bvc(\cY^{\top})]^{\top}[(v_m\otimes I_M)\otimes_b(v_m\otimes I_M)]\vc(X_m)\nonumber\\
&=\sum_{m=1}^N\tr\left(\text{unbvec}\{(v_m\otimes I_M)\otimes_b(v_m\otimes I_M)\bvc(X_m)\}\cY\right)\nonumber\\
&=\sum_{m=1}^N\tr\left((v_m\otimes I_M)X_m(v^{\top}_m\otimes I_M)\cY\right)\nonumber\\
&=\mu^2\sum_{m=1}^N(1-\mu\eta\lambda_m)^2\tr\left((v^{\top}_m\otimes I_M)\cS_\eta(v_m\otimes I_M)X_m\right)\nonumber\\
&=\frac{\mu}{2}\sum_{m=1}^N(1-\mu\eta\lambda_m)^2\tr\left(\left(\sum_{k=1}^N[v_m]^2_kR_{s,k,\eta}\right)\left((1-\mu\eta\lambda_m)^2\left(\sum_{k=1}^N[v_m]^2_kH_{k,\eta}\right)+\frac{\eta}{2}\lambda_m(2-\mu\eta\lambda_m)I\right)^{-1}\right)\nonumber\\
&=\frac{\mu}{2}\sum_{m=1}^N\tr\left(\left(\sum_{k=1}^N[v_m]^2_kH_{k,\eta}+\frac{\eta\lambda_m(2-\mu\eta\lambda_m)}{2(1-\mu\eta\lambda_m)^2}I\right)^{-1}\left(\sum_{k=1}^N[v_m]^2_kR_{s,k,\eta}\right)\right).
\end{align}
Substituting into~\eqref{eq: relation expression} and~\eqref{eq: steady-state network performance}, we conclude:
\begin{equation}
\label{eq: final expression for the network MSE}
\begin{split}
\limsup_{i\rightarrow\infty}\frac{1}{N}\expec\|\cw^o_{\eta}-\bcw_i\|^2&=\frac{\mu}{2N}\sum_{m=1}^N\tr\left(\left(\sum_{k=1}^N[v_m]^2_kH_{k,\eta}+\frac{\eta\lambda_m(2-\mu\eta\lambda_m)}{2(1-\mu\eta\lambda_m)^2}I\right)^{-1}\left(\sum_{k=1}^N[v_m]^2_kR_{s,k,\eta}\right)\right)\\
&\qquad\qquad+\frac{O(\mu)}{(O(1)+O(\eta))}+O(\mu^{1+{{\theta}}_m}).
\end{split}
\end{equation}
Now, according to definition~\eqref{eq: network MSD performance}, dividing~\eqref{eq: final expression for the network MSE} by $\mu$ and computing the limit as $\mu\rightarrow 0$, we arrive at expression~\eqref{eq: final expression for the network MSD alternative} for the network MSD.
\end{appendices}

\bibliographystyle{IEEEbib}
\bibliography{reference}

\begin{thebibliography}{10}

\bibitem{nassif2018distributed}
R.~Nassif, S.~Vlaski, and A.~H. Sayed,
\newblock ``Distributed inference over multitask graphs under smoothness,''
\newblock in {\em Proc. IEEE International Workshop on Signal Processing
  Advances in Wireless Communications}, Kalamata, Greece, Jun. 2018.

\bibitem{nassif2018diffusion}
R.~Nassif, S.~Vlaski, C.~Richard, and A.~H. Sayed,
\newblock ``Learning over multitask graphs -- {P}art {I}: {S}tability
  analysis,''
\newblock {\em Submitted for publication}, Nov. 2019.

\bibitem{bertsekas1997new}
D.~P. Bertsekas,
\newblock ``A new class of incremental gradient methods for least squares
  problems,''
\newblock {\em SIAM J. Optim.}, vol. 7, no. 4, pp. 913--926, 1997.

\bibitem{olfati2007consensus}
R.~Olfati-Saber, J.~A. Fax, and R.~M. Murray,
\newblock ``Consensus and cooperation in networked multi-agent systems,''
\newblock {\em Proc. IEEE}, vol. 95, no. 1, pp. 215--233, 2007.

\bibitem{nedic2009distributed}
A.~Nedic and A.~Ozdaglar,
\newblock ``Distributed subgradient methods for multi-agent optimization,''
\newblock {\em IEEE Trans. Autom. Control}, vol. 54, no. 1, pp. 48--61, 2009.

\bibitem{dimakis2010gossip}
A.~G. Dimakis, S.~Kar, J.~M.~F. Moura, M.~G. Rabbat, and A.~Scaglione,
\newblock ``Gossip algorithms for distributed signal processing,''
\newblock {\em Proc. IEEE}, vol. 98, no. 11, pp. 1847--1864, 2010.

\bibitem{ram2010distributed}
S.~S. Ram, A.~Nedi{\'c}, and V.~V. Veeravalli,
\newblock ``Distributed stochastic subgradient projection algorithms for convex
  optimization,''
\newblock {\em J. Optim. Theory Appl.}, vol. 147, no. 3, pp. 516--545, 2010.

\bibitem{chen2013distributed}
J.~Chen and A.~H. Sayed,
\newblock ``Distributed {P}areto optimization via diffusion strategies,''
\newblock {\em IEEE J. Sel. Topics Signal Process.}, vol. 7, no. 2, pp.
  205--220, 2013.

\bibitem{sayed2014adaptation}
A.~H. Sayed,
\newblock ``Adaptation, learning, and optimization over networks,''
\newblock {\em Foundations and Trends in Machine Learning}, vol. 7, no. 4-5,
  pp. 311--801, 2014.

\bibitem{chen2015learning}
J.~Chen and A.~H. Sayed,
\newblock ``On the learning behavior of adaptive networks -- {P}art {I}:
  {T}ransient analysis,''
\newblock {\em IEEE Trans. Inf. Theory}, vol. 61, no. 6, pp. 3487--3517, Jun.
  2015.

\bibitem{chen2015learning2}
J.~Chen and A.~H. Sayed,
\newblock ``On the learning behavior of adaptive networks -- {P}art {II}:
  {P}erformance analysis,''
\newblock {\em IEEE Trans. Inf. Theory}, vol. 61, no. 6, pp. 3518--3548, Jun.
  2015.

\bibitem{sayed2014adaptive}
A.~H. Sayed,
\newblock ``Adaptive networks,''
\newblock {\em Proc. IEEE}, vol. 102, no. 4, pp. 460--497, Apr. 2014.

\bibitem{vlaski2016diffusion}
S.~Vlaski, L.~Vandenberghe, and A.~H. Sayed,
\newblock ``Diffusion stochastic optimization with non-smooth regularizers,''
\newblock in {\em Proc. Int. Conf. Acoust., Speech, Signal Process.}, Shanghai,
  China, Mar. 2016, pp. 4149--4153.

\bibitem{platachaves2017heterogeneous}
J.~Plata-Chaves, A.~Bertrand, M.~Moonen, S.~Theodoridis, and A.~M. Zoubir,
\newblock ``Heterogeneous and multitask wireless sensor networks --
  {A}lgorithms, applications, and challenges,''
\newblock {\em IEEE J. Sel. Topics Signal Process.}, vol. 11, no. 3, pp.
  450--465, Apr. 2017.

\bibitem{hassani2017multi}
A.~Hassani, J.~Plata-Chaves, M.~H. Bahari, M.~Moonen, and A.~Bertrand,
\newblock ``Multi-task wireless sensor network for joint distributed
  node-specific signal enhancement, {LCMV} beamforming and {DOA} estimation,''
\newblock {\em IEEE J. Sel. Topics Signal Process.}, vol. 11, no. 3, pp.
  518--533, 2017.

\bibitem{chen2014multitask}
J.~Chen, C.~Richard, and A.~H. Sayed,
\newblock ``Multitask diffusion adaptation over networks,''
\newblock {\em IEEE Trans. Signal Process.}, vol. 62, no. 16, pp. 4129--4144,
  2014.

\bibitem{nassif2016proximal}
R.~Nassif, C.~Richard, A.~Ferrari, and A.~H. Sayed,
\newblock ``Proximal multitask learning over networks with sparsity-inducing
  coregularization,''
\newblock {\em IEEE Trans. Signal Process.}, vol. 64, no. 23, pp. 6329--6344,
  2016.

\bibitem{eksin2012distributed}
C.~Eksin and A.~Ribeiro,
\newblock ``Distributed network optimization with heuristic rational agents,''
\newblock {\em IEEE Trans. Signal Process.}, vol. 60, no. 10, pp. 5396--5411,
  Oct. 2012.

\bibitem{hallac2015network}
D.~Hallac, J.~Leskovec, and S.~Boyd,
\newblock ``Network {L}asso: {C}lustering and optimization in large graphs,''
\newblock in {\em Proc. ACM SIGKDD}, Sydney, Australia, Aug. 2015, pp.
  387--396.

\bibitem{szurley2015distributed}
J.~Szurley, A.~Bertrand, and M.~Moonen,
\newblock ``Distributed adaptive node-specific signal estimation in
  heterogeneous and mixed-topology wireless sensor networks,''
\newblock {\em Signal Processing}, vol. 117, pp. 44--60, 2015.

\bibitem{platachaves2015distributed}
J.~Plata-Chaves, N.~Bogdanovi{\'c}, and K.~Berberidis,
\newblock ``Distributed diffusion-based {LMS} for node-specific adaptive
  parameter estimation,''
\newblock {\em {IEEE} Trans. Signal Process.}, vol. 63, no. 13, pp. 3448--3460,
  2015.

\bibitem{alghunaim2017decentralized}
S.~A. Alghunaim, K.~Yuan, and A.~H. Sayed,
\newblock ``Decentralized exact coupled optimization,''
\newblock in {\em Proc. Ann. Allerton Conf. on Communication, Control, and
  Computing}, Illinois, USA, 2017, pp. 338--345.

\bibitem{nassif2017diffusion}
R.~Nassif, C.~Richard, A.~Ferrari, and A.~H. Sayed,
\newblock ``Diffusion {LMS} for multitask problems with local linear equality
  constraints,''
\newblock {\em {IEEE} Trans. Signal Process.}, vol. 65, no. 19, pp. 4979--4993,
  2017.

\bibitem{chen2015diffusion}
J.~Chen, C.~Richard, and A.~H. Sayed,
\newblock ``Diffusion {LMS} over multitask networks,''
\newblock {\em {IEEE} Trans. Signal Process.}, vol. 63, no. 11, pp. 2733--2748,
  2015.

\bibitem{chen2014diffusion}
J.~Chen, C.~Richard, A.~O. Hero, and A.~H. Sayed,
\newblock ``Diffusion {LMS} for multitask problems with overlapping hypothesis
  subspaces,''
\newblock in {\em Proc. IEEE Int. Workshop Mach. Learn. Signal Process.},
  Reims, France, Sep. 2014, IEEE, pp. 1--6.

\bibitem{kekatos2013distributed}
V.~Kekatos and G.~B. Giannakis,
\newblock ``Distributed robust power system state estimation,''
\newblock {\em IEEE Trans. Signal Process.}, vol. 28, no. 2, pp. 1617--1626,
  2013.

\bibitem{zhou2004regularization}
D.~Zhou and B.~Sch{\"o}lkopf,
\newblock ``A regularization framework for learning from graph data,''
\newblock in {\em Proc. ICML Workshop on Statistical Relational Learning and
  Its Connections to Other Fields}, 2004, vol.~15, pp. 67--68.

\bibitem{shuman2013emerging}
D.~I. Shuman, S.~K. Narang, P.~Frossard, A.~Ortega, and P.~Vandergheynst,
\newblock ``The emerging field of signal processing on graphs: {E}xtending
  high-dimensional data analysis to networks and other irregular domains,''
\newblock {\em IEEE Signal Process. Mag.}, vol. 30, no. 3, pp. 83--98, May
  2013.

\bibitem{chung1997spectral}
F.~R.~K. Chung,
\newblock {\em Spectral {G}raph {T}heory},
\newblock American Mathematical Society, 1997.

\bibitem{ando2006learning}
R.~K. Ando and T.~Zhang,
\newblock ``Learning on graph with {L}aplacian regularization,''
\newblock in {\em Proc. Advances in neural information processing systems},
  Canada, Dec. 2006, pp. 25--32.

\bibitem{dong2016learning}
X.~Dong, D.~Thanou, P.~Frossard, and P.~Vandergheynst,
\newblock ``Learning {L}aplacian matrix in smooth graph signal
  representations,''
\newblock {\em {IEEE} Trans. Signal Process.}, vol. 64, no. 23, pp. 6160--6173,
  2016.

\bibitem{chen2017bias}
P.-Y. Chen and S.~Liu,
\newblock ``Bias-variance tradeoff of graph {L}aplacian regularizer,''
\newblock {\em IEEE Signal Process. Lett.}, vol. 24, no. 8, pp. 1118--1122,
  2017.

\bibitem{polyak1987introduction}
B.~T. Polyak,
\newblock ``Introduction to {O}ptimization,''
\newblock {\em Optimization Software, New York}, 1987.

\bibitem{koning1991block}
R.~H. Koning, H.~Neudecker, and T.~Wansbeek,
\newblock ``Block {K}ronecker products and the vecb operator,''
\newblock {\em Linear algebra and its applications}, vol. 149, pp. 165--184,
  Apr. 1991.

\bibitem{bernstein2005matrix}
D.~S. Bernstein,
\newblock {\em Matrix {M}athematics: {T}heory, {F}acts, and {F}ormulas with
  {A}pplication to {L}inear {S}ystems {T}heory},
\newblock Princeton University Press, 2005.

\bibitem{isserlis1918formula}
L.~Isserlis,
\newblock ``On a formula for the product-moment coefficient of any order of a
  normal frequency distribution in any number of variables,''
\newblock {\em Biometrika}, vol. 12, no. 1/2, pp. 134--139, Nov. 1918.

\bibitem{shuman2018distributed}
D.~I. Shuman, P.~Vandergheynst, D.~Kressner, and P.~Frossard,
\newblock ``Distributed signal processing via {C}hebyshev polynomial
  approximation,''
\newblock {\em {IEEE} Trans. Signal Inf. Process. Netw.}, 2018.

\bibitem{sandryhaila2013discrete}
A.~Sandryhaila and J.~M.~F. Moura,
\newblock ``Discrete signal processing on graphs,''
\newblock {\em IEEE Trans. Signal Process.}, vol. 61, no. 7, pp. 1644--1656,
  Apr. 2013.

\bibitem{ortega2018graph}
A.~Ortega, P.~Frossard, J.~Kova{\v{c}}evi{\'c}, J.~M.~F. Moura, and
  P.~Vandergheynst,
\newblock ``Graph signal processing: {O}verview, challenges, and
  applications,''
\newblock {\em Proc. IEEE}, vol. 106, no. 5, pp. 808--828, 2018.

\bibitem{kailath1980linear}
T.~Kailath,
\newblock {\em Linear {S}ystems},
\newblock Prentice-Hall, Englewood Cliffs, NJ, 1980.

\bibitem{boyd2004convex}
S.~Boyd and L.~Vandenberghe,
\newblock {\em Convex {O}ptimization},
\newblock Cambridge University Press, NY, 2004.

\end{thebibliography}

\end{document}